\documentclass[runningheads]{llncs}
\usepackage{amssymb}
\usepackage{amsmath,amssymb}
\usepackage{lscape, rotating}
\usepackage{graphicx}
\usepackage{extarrows}

\newtheorem{no}{\textbf{Note}}

\newtheorem{rem}{\textbf{Remark}}

\newtheorem{exa}{\textbf{Example}}

\begin{document}

\pagestyle{plain}

\mainmatter

\title{On Secrecy Capacity of Minimum Storage Regenerating Codes
\protect\footnote{$^1$Kun Huang and Ming Xian are with State Key Laboratory of Complex Electromagnetic Environment Effects on Electronics and Information System, National University of Defense Technology, Changsha, 410073, China \email{(khuangresearch923@gmail.com; qwertmingx@tom.com)}.\\
$^2$ Udaya Parampalli is with Department of Computing and Information Systems, University of Melbourne,  VIC 3010, Australia \email{(udaya@unimelb.edu.au)}.\\
}
}
\titlerunning{}

\author{
Kun Huang\inst{1},\ Udaya Parampalli\inst{2}, \and\ Ming Xian\inst{1}
}

\authorrunning{}

\institute{}

\maketitle

\begin{abstract}
In this paper, we revisit the problem of characterizing the secrecy capacity of minimum storage regenerating (MSR) codes under the passive $(l_1,l_2)$-eavesdropper model, where the eavesdropper has access to data stored on $l_1$ nodes and the repair data for an additional $l_2$ nodes. We study it from the information-theoretic perspective. First, some general properties of MSR codes as well as a simple and generally applicable upper bound on secrecy capacity are given. Second, a new concept of \emph{stable} MSR codes is introduced, where the stable property is shown to be closely linked with secrecy capacity. Finally, a comprehensive and explicit result on secrecy capacity in the linear MSR scenario is present, which generalizes all related works in the literature and also predicts certain results for some unexplored linear MSR codes.

{\bf Key Words: } MSR Codes, Repair Data, Secrecy Capacity, Upper Bounds.
\end{abstract}

\section{INTRODUCTION}
Distributed storage systems (DSSs) are an essential part of large scale data storage systems required for many new emerging distributed networking applications such as social networking, video sharing, peer to peer networking and large scale data centres. As is common in such storage systems, redundancy is indispensably introduced to ensure reliability and availability owing to frequent node failures. The main approaches to introduce redundancy in DSSs are through replication, erasure codes, and more recently using regenerating codes \cite{Re:A.Dimakis}. Erasure codes in general can achieve higher reliability for the same level of redundancy when compared to the schemes that provide replication \cite{Re:H.Weatherspoon}. Regenerating codes are a recent innovation of erasure codes that has efficient performance on repair of failed nodes in DSSs \cite{Re:A.G.Dimakis1}.

\subsection{Regenerating Codes.}

Regenerating codes \cite{Re:A.Dimakis} are a family of maximal distance separable (MDS) codes determined by a tradeoff between the amount of storage per node and the repair bandwidth. In the framework of regenerating codes, an encoded data file is split into $n\alpha$ symbols and then dispersed across $n$ nodes, where all the symbols are drawn from a finite field $\mathbb{F}_q$ and each node stores a collection of $\alpha$ symbols. The dispersing manner requires that any data collector can retrieve the original data message by connecting to any $k$ out of $n$ nodes. The node repair can be accomplished by permitting a new node to connect to any $d$ helper nodes from the surviving $(n-1)$ nodes by downloading $\beta\leq \alpha$ symbols from each node. In the literature, a regenerating code is represented by a parameter set $\{n,k,d,\alpha,\beta,B\}$, where $B$ is the size of original data message and $d\beta$ is the total amount of data transferred for node repair that is termed repair bandwidth.

The cut-set bound based on the concept of information flow \cite{Re:R.Ahlswede} requires that the parameters of a regenerating code must necessarily satisfy:
\begin{equation}\label{cut}
B\leq \sum_{i=1}^{k}\min\{\alpha,(d-i+1)\beta\}.
\end{equation}
In \cite{Re:A.Dimakis}, Dimakis et al derive the above tradeoff between the per node storage $\alpha$ at each node and repair bandwidth $d\beta$. The codes that can achieve this tradeoff curve are called \emph{optimal} regenerating codes. Two extreme points on this tradeoff curve are of particular concern, namely, minimum bandwidth regenerating (MBR) point and minimum storage regenerating (MSR) point, respectively representing codes with the least repair bandwidth and ones with the least per node storage. As shown in \cite{Re:A.Dimakis}, the parameters of MBR and MSR codes are given by:
\begin{equation}\label{optimal}
\left\{\begin{aligned}
&(\alpha_{\textbf{MSR}},\beta_{\textbf{MSR}})=(\frac{B}{k}, \frac{B}{k(d-k+1)})\\
&(\alpha_{\textbf{MBR}},\beta_{\textbf{MBR}})=(\frac{2dB}{k(2d-k+1)}, \frac{2B}{k(2d-k+1)})
\end{aligned}\right.
\end{equation}
In the literature, three repair models are considered: functional repair, exact repair, and exact repair of systematic nodes \cite{Re:A.G.Dimakis1}. Exact repair can regenerate the exact replicas of the lost data in the failed nodes and thus is preferred in practical systems \cite{Re:C.Huang}. In the exact repair scenario, Shah et al in \cite{Re:Shah.N.B} demonstrate that most interior points on the storage-bandwidth tradeoff curve are not achievable. For those possibly reachable interior points, constructions of codes are rare \cite{Re:C. Tian,Re:T. Ernvall}. In addition, Duursma in \cite{Re:Duursma1,Re:Duursma2} derive some new outer bounds for regenerating codes with exact repair.

Up to now, several constructions with the exact repair property for MBR and MSR codes have been proposed. In \cite{Re:K.Rashmi}, Rashmi et al employ product matrix to construct MBR codes for all parameters and MSR codes with $\{d\geq 2k-2\}$. In the MSR scenario, significant progress have been made. From the overall perspective, there are two classes of MSR codes, i.e., the scalar MSR codes with $\{\beta=1\}$ \cite{Re:K.Rashmi,Re:C.Suh,Re:Y.Wu1,Re:K.V.Rashmi,Re:Rashmi,Re:N. B. Shah} and vector MSR codes with $\{\beta=(n-k)^{x}\}$ where $x\geq 1$ \cite{Re:Z. Wang,Re:D. S. Papailiopoulos,Re:I.Tamo,Re:Z. Wang1,Re:V. R. Cadambe,Re:V. R. Cadambe1,Re:G. K. Agarwal,Re:Y.S.Han,Re:J.Li}. Many of these constructions are established on interference alignment. As explained in \cite{Re:N. B. Shah}, interference alignment is the necessity of constructing linear scalar MSR codes and these linear scalar MSR codes only exist when $d\geq 2k-2$. From another point of view, this existing restriction exactly corresponds to the low rate regime, i.e., $\frac{k}{n}\leq \frac{1}{2n}+\frac{1}{2}$. As for the high rate codes with $\{\frac{k}{n}>\frac{1}{2}\}$, vector MSR codes are available as they are free from the parameter constraints of $(n,k)$. However, many of these vector codes only allow efficient repair of systematic nodes \cite{Re:I.Tamo,Re:Z. Wang1,Re:V. R. Cadambe,Re:V. R. Cadambe1,Re:G. K. Agarwal,Re:Y.S.Han,Re:J.Li}, such as Zigzag codes \cite{Re:I.Tamo}. In \cite{Re:Z. Wang,Re:D. S. Papailiopoulos}, the authors present vector MSR codes allowing efficient repair for parity nodes as well, where the code given in \cite{Re:Z. Wang} is a variant of Zigzag code.

In addition to repair efficiency, there are many other design features required by DSSs such as security
\cite{Re:S.Pawar,Re:N.B.Shah,Re:S.Goparaju,Re:Rawat.A.S,Re:Tandon,Re:Koyluoglu}, local-repairability \cite{Re:Rawat.A.S,Re:N.Prakash,Re:P.Gopalan,Re:D. S. Papailiopoulos1}, optimality of updating
\cite{Re:I.Tamo,Re:Y.S.Han,Re:J.Li}, etc. Our concern in this paper is on securing DSSs against eavesdroppers attempting to obtain any knowledge of the original data.

\subsection{Secure Regenerating Codes.}

Since the nodes of DSSs are widely spread across the network, individual nodes may be compromised and as a result the data stored is vulnerable to eavesdropping. There are mainly two kinds of attacker models considered in the literature: passive eavesdropper model and active eavesdropper model \cite{Re:H.Delfs}. Compared to the former, active eavesdropper can modify the data or even inject new data into the compromised nodes. Our eavesdropper model considered in this paper is the passive one as given in \cite{Re:N.B.Shah}. In this model, eavesdropper has access to the data stored on $l_1$ nodes as well as the repair data for an additional $l_2$ nodes. Here, we only consider the situation of exact repair\footnote{Functional repair scheme requires ceaselessly updating the data stored in nodes undergoing repair, which may leak substantial linear combinations of data to eavesdroppers and enable the eavesdroppers to retrieve the original data just by solving the linear equations. This is another reason why exact repair is superior to functional repair.}.

{\bf Related work:} The issue of designing secure regenerating codes against eavesdropping was firstly addressed in \cite{Re:S.Pawar} and \cite{Re:N.B.Shah}. The authors in \cite{Re:S.Pawar} considered the initial setting that an eavesdropper observes the contents of $l<k$ nodes of the storage system, and analyzed the regenerating code's secrecy capacity (i.e., the maximal file size that can be securely stored). An upper bound of the secrecy capacity and a secure MBR code that can attain this bound are proposed in \cite{Re:S.Pawar}. Extending the initial eavesdropper setting \cite{Re:S.Pawar}, authors in \cite{Re:N.B.Shah} modeled the eavesdropper as one obtaining access to the data stored on $l_1$ nodes as well as the repair data for an additional $l_2$ nodes, with $l_1+l_2<k$. The secure product-matrix-based MBR coding scheme proposed in \cite{Re:N.B.Shah} can achieve the bound derived in \cite{Re:S.Pawar} with $l=l_1+l_2$. Achievability of the bound for secure product-matrix-based MBR codes in \cite{Re:N.B.Shah} can be attributed to the fact that the repair bandwidth $d\beta$ equals to per node storage $\alpha$ in the MBR scenario. In other words, the $(l_1,l_2)$-eavesdropper cannot obtain any extra information other than the contents of $l=l_1+l_2$ nodes in the MBR scenario. Hence, under the $(l_1,l_2)$-eavesdropper model, the upper bound in \cite{Re:S.Pawar} still holds for the secure file size $B^{(s)}$ of MBR codes:
\begin{equation}\label{initial bound}
B^{(s)}\leq \sum_{i=l+1}^{k}\min\{\alpha,(d-i+1)\beta\},
\end{equation}
where $l=l_1+l_2$. Authors in \cite{Re:N.B.Shah} further considered the design of secure MSR codes based on product-matrix codes, but the secure MSR coding scheme is only capable of storing $(k-l_1-l_2)(\alpha-l_2\beta)$-sized secure files, which reaches the bound (\ref{initial bound}) only when $l_2=0$. The intuition here indicates that the $(l_1,l_2)$-eavesdropper can obtain more information than the contents of $(l_1+l_2)$ nodes in the MSR scenario, as the repair bandwidth $d\beta$ is larger than  $\alpha=(d-k+1)\beta$ that is the amount of data stored on each of those $l_2$ nodes. As mentioned in \cite{Re:N.B.Shah}, it was unknown yet whether such a secure MSR code is still optimal when $l_2\geq 1$. Since then, characterization of the secrecy capacity for MSR codes is considered to be open under $(l_1,l_2>0)$-eavesdropper model.

Recently, the authors in \cite{Re:S.Goparaju} and \cite{Re:Rawat.A.S} employ the technique of linear subspace intersection and then derive new upper bounds on secrecy capacity for linear MSR codes. Zigzag code \cite{Re:I.Tamo} and its variant \cite{Re:Z. Wang} are shown to achieve these bounds through pre-coding of maximum rank distance (MRD) code \cite{Re:R. M. Roth,Re:Gabidulin}. The bound given in \cite{Re:Rawat.A.S} auxiliarily implies that the product-matrix-based secure MSR code proposed in \cite{Re:N.B.Shah} is also optimal for $l_2=1$. Regarding the bound given in \cite{Re:S.Goparaju}, it is actually an extension of the one in \cite{Re:Rawat.A.S}, since the bound in \cite{Re:S.Goparaju} matches to that in \cite{Re:Rawat.A.S} when $l_2\leq2$.

In another parallel research area, towards two separate eavesdropper models with $(l_1,l_2=0)$ and $(l_1=0,l_2)$, the authors in \cite{Re:Tandon} study the secure storage-vs-repair-bandwidth tradeoff, where they respectively derive new outer bounds on secrecy capacity for a general parameter set and some specific parameter sets. Therein, they show that in the presence of $(l_1=0,l_2)$-eavesdropper, these new bounds strictly improve upon the existing cutset-based bounds presented in \cite{Re:S.Pawar} and the MBR point is the only efficient point that can achieve these specific-parameter-based bounds. Under the above background of $(l_1,l_2)$-eavesdropper model, our focus herein is dedicated to studying the secrecy capacity solely at the MSR point\footnote{It is shown in \cite{Re:Tandon} that for certain parameters, secure codes operating at the MBR points actually have better ``storage" (i.e., the maximal file size that can be securely stored, or just termed secrecy capacity) rate than codes operating at the MSR points. In this sense, it appears that secure MSR codes lose the feature of optimal storage, while the original notion of MSR codes under the non-secure setting shall be optimal in storage rate as displayed in \cite{Re:A.Dimakis}. Throughout this paper, we still use the term MSR points (or codes) to only signify the fact that $\alpha$ and $\beta$ satisfy the relationship $\alpha=(d-k+1)\beta$ like the MBR points termed in \cite{Re:Tandon} that require $\alpha=d\beta$. Essentially, each node in the secure MSR codes still stores $\alpha_{\textbf{MSR}}$ symbols and transmits $\beta_{\textbf{MSR}}$ symbols for repairing failed nodes, which just need to replace with some randomness.}.

{\bf Contributions:} In this work, we first carefully review the method of determining regenerating codes considered in \cite{Re:A.Dimakis} and the information-theoretic technique used in \cite{Re:Shah.N.B}. Therein, we find that the $\alpha$ symbols stored in any node or the $\beta$ symbols contained in any single set of repair data for the \emph{optimal} regenerating codes are in fact mutually independent and uniformly distributed inside themselves. It not only indicates that entropy of any symbol involved reaches the maximal value $1$, but also signifies that
entropy of the $\alpha$ symbols in any node and entropy of the $\beta$ symbols in any single repair data all attain the maximal-integer-value $\alpha$ and $\beta$ respectively. Thereafter, we recognize that the concepts of uniform distribution and independence between symbols in information theory \cite{Re:Cover.T.M} exactly correspond to those of permutation polynomial and orthogonal system in finite fields \cite{Re:R. Lidl} respectively. Using these two theories in finite fields, we demonstrate that the joint entropy of symbols included in multiple sets of repair data in the nonlinear context may be a non-integer value while it must be an integer in the linear context, which will be used to
investigate the secrecy capacity of linear MSR codes.

Then, we turn to study the inherent features of MSR codes from the information-theoretic perspective, where the data stored in storage nodes and transferred by helper nodes during repair are considered as random variables. Based on the basic reconstruction and regeneration properties of MSR codes with $\{n=d+1,k,d,\alpha,\beta\}$, we derive two useful properties: (i) the repair data sent from disjoint sets of nodes to a failed node are mutually independent, and (ii) given the contents of a node and the repair data from any $k-1$ nodes, the repair data from the remaining $d-k+1$ nodes are deterministic. Combining the two new properties with a universal upper bound on secrecy capacity for any \emph{optimal} regenerating code with $\{n=d+1,k,d,\alpha,\beta\}$, we derive a simple and generally applicable upper bound on secrecy capacity for any MSR code with $\{n=d+1,k,d,\alpha,\beta\}$. As for the MSR codes with $\{n>d+1,k,d,\alpha,\beta\}$, we introduce a new concept of ``\emph{stable}" MSR codes, which require that repair data transmitted from any node $i$ to any failed node $j$ is independent of the choice of the set of helper nodes including the same node $i$. Therein, we show this stable property is the equivalent condition of secrecy capacity between any MSR code with $\{n>d+1,k,d,\alpha,\beta\}$ and its truncated one with $\{n=d+1,k,d,\alpha,\beta\}$. It should be noted that the product-matrix-based MSR code given in \cite{Re:K.Rashmi} is a \emph{stable} MSR code.

Finally, we converge back to the linear MSR codes with parameter set $\{n=d+1,k,d,\alpha,\beta\}$, where those aforementioned upper bounds on secrecy capacity actually can always be achieved through the pre-coding of maximum rank distance (MRD) code \cite{Re:R. M. Roth,Re:Gabidulin} as applied in \cite{Re:Rawat.A.S,Re:Koyluoglu}. Based on the fact that joint entropy of multiple sets of repair data is an integer, we fully characterize the secrecy capacity of linear MSR codes in the category where $1\leq\beta<\frac{d-k+1}{l_2-1}$. A consequence of this result when $\beta =1$ naturally establishes the optimality of product-matrix-based secure MSR codes whenever $l_1+l_2\leq k-1$ and $l_2\leq \min\{k-1,d-k+1\}$, which completely resolves the question raised in \cite{Re:N.B.Shah}. Note that product-matrix-based MSR code given in \cite{Re:K.Rashmi} is a scalar MSR code, i.e., it is built on $\beta=1$. In the other category where $\beta\geq\frac{d-k+1}{l_2-1}$, we give new upper bounds on secrecy capacity, which are in fact improved generalization of the results given in \cite{Re:S.Goparaju,Re:Rawat.A.S}. Thereafter, we find that all the aforementioned results also apply to systematic MSR codes with only repair data of systematic nodes eavesdropped. By putting all together, we eventually present a comprehensive and explicit result on secrecy capacity for linear MSR codes, which closely depends on the value of $\beta$. This final outcome also predicts certain results on secrecy capacity for some unexplored linear MSR codes. As an illustration and comparison, Table \ref{Tab:BoundComparision} summaries the study progresses on secrecy capacity for linear MSR codes, wherein it should be noted that the bound in \cite{Re:Tandon} cannot be reached for MSR codes.

\begin{table*}[ht]
 \tabcolsep 0pt
\begin{center}
\def\temptablewidth{0.961\textwidth}
\caption{Secrecy Capacity of Linear MSR Codes under $(l_1,l_2)$-Eavesdropper Model}\label{Tab:BoundComparision}
{\rule{\temptablewidth}{1pt}}
\begin{tabular}{|c|c|}
\quad Citation        \quad     &  \quad Corresponding Results   \quad     \\ \hline
\quad Pawar et al\cite{Re:S.Pawar}  \quad     &  \quad $B^{(s)}\leq (k-l_1-l_2)\alpha$, optimal only when $l_2=0$ \quad     \\ \hline
\quad Tandon et al    \cite{Re:Tandon} \quad     &  \quad $B^{(s)}\leq(k-l_2)(1-\frac{1}{d})\alpha$, for $n=d+1$, $l_1=0$ and $1\leq l_2<k$ \quad     \\ \hline
\quad Shah et al    \cite{Re:N.B.Shah} \quad     &  \quad $B^{(s)}=(k-l_1-l_2)(\alpha-l_2\beta)$, for product-matrix-based MSR codes \quad     \\ \hline
\quad Rawat et al\cite{Re:Rawat.A.S}    \quad       & \quad  $B^{(s)}\leq (k-l_1-l_2)\big(\alpha-\theta(\beta,l_2)\big)$, where
$\theta(\beta,l_2)=\left\{\begin{aligned}
&\beta, \quad for\quad l_2=1\\
&2\beta-\frac{\beta}{d+1-k}, \quad for\quad l_2=2
\end{aligned}\right.$                       \quad      \\ \hline
\quad Goparaju et al\cite{Re:S.Goparaju}    \quad   & \quad  $B^{(s)}\leq (k-l_1-l_2)(1-\frac{1}{n-k})^{l_2}\alpha$, where $n=d+1$  \quad                         \\ \hline
\quad      & \quad  $\underbrace{B^{(s)}= (k-l_1-l_2)\big(\alpha-\pi(\beta,l_2)\big)}$, wherein\\
\quad This paper    \quad & $\pi(\beta,l_2)$:
$\left\{\begin{array}{ll}
=l_2\beta, &if \quad l_2\leq t, \beta<\frac{d-k+1}{t-1};\\
\geq t\beta+\beta(d-k-t+1)\big{[}1-(\frac{d-k}{d-k+1})^{e} \big{]}, &if \quad l_2=t+e, \frac{d-k+1}{t}\leq\beta<\frac{d-k+1}{t-1},
\end{array}\right.$\\
 & where $1\leq t\leq d-k+1$ and $e\geq 1$. This also can be referenced from our formula (\ref{explicit})   \\
\end{tabular}
{\rule{\temptablewidth}{1pt}}
\end{center}
\end{table*}

{\bf Organization:} Section 2 gives preliminaries consisting of some basic definitions in information theory, notation used in this paper, some results from theory of finite field and a universal upper bound under the $(l_1,l_2)$-eavesdropper model. Section 3 presents some new results for general MSR codes mainly including some general properties, some generally applicable upper bounds on secrecy capacity and the new concept of \emph{stable} MSR codes. Section 4 exhibits the comprehensive and explicit result on secrecy capacity for linear MSR codes. Section 5 concludes this paper.

\section{PRELIMINARIES}

In this section, some basic concepts related to information theory are quoted, which will be used in high frequency later. Then, we describe the system model of MSR codes from the information-theoretic perspective. Subsequently, we introduce the theory on permutation polynomial in finite fields, which can be regarded as a new way to understand the construction of \emph{optimal} regenerating codes. At last, we present a universal upper bound on secrecy capacity under the $(l_1,l_2)$-eavesdropper model.

\subsection{Information Entropy}

\begin{definition}\cite{Re:Cover.T.M}\label{entropy}(Entropy of A Random Variable $X$): The entropy of a discrete random variable $X$ with probability distribution $p_X(x)$ is
\begin{equation}
H(X)=-\sum_x{p(x)\log{p(x)}}.
\end{equation}
The entropy measures the expected uncertainty in $X$. It must be that $H(X)\geq0$, meaning entropy is always non-negative and $H(X) = 0$ iff $X$ is deterministic. In addition, when $X$ is uniformly distributed (i.e., $p(x)=\frac{1}{q}$ where $q$ is the total number of the events of $X$), $H(X)$ achieves the maximum value $\log q$. Normally, the base of logarithm can be specified to $q$. In this case, it must be that $H(X)\leq 1$ and $H(X)=1$ iff $X$ is uniformly distributed.
\end{definition}

\begin{definition}\cite{Re:Cover.T.M}\label{joint}(Joint Entropy and Conditional Entropy): Joint entropy between two random variables $X$ and $Y$, and conditional entropy of $Y$ given a random variable $X$ are respectively
\begin{equation}
\left\{\begin{aligned}
&H(X,Y)=-E_{p(x,y)}[\log{p(X,Y)}]=-\sum_x\sum_y{p(x,y)\log{p(x,y)}}\\
&H(Y|X)=-E_{p(x,y)}[log{p(Y|X)}]=-\sum_xp(x)H(Y|X=x)
\end{aligned}\right.
\end{equation}
Besides, joint and conditional entropy provide a natural calculus: $H(X,Y)=H(X)+H(Y|X)$.

\end{definition}

\begin{definition}\cite{Re:Cover.T.M}(Mutual Information and Conditional Mutual Information): The mutual information between $X$ and $Y$, and the conditional mutual information between $X$ and $Y$ given another random variable $Z$ are respectively given by:
\begin{equation}
\left\{\begin{aligned}
&I(X;Y)=H(X)-H(X|Y)\\
&I(X;Y|Z)=H(X|Z)-H(X|Y,Z)\\
\end{aligned}\right.
\end{equation}
\end{definition}

\begin{definition}\cite{Re:Cover.T.M}(Chain Rules): Chain rules for entropy and mutual information are:
\begin{equation}
\left\{\begin{aligned}
&H(X_1,\cdots,X_n)=\sum_{i=1}^nH(X_i|X_{i-1},\cdots,X_1)\\
&I(X_1,\cdots,X_n;Y)=\sum_{i=1}^nI(X_i;Y|X_{i-1},X_{i-2},\cdots,X_1)
\end{aligned}\right.
\end{equation}
\end{definition}

\begin{lemma}\label{helpful}
Based on these definitions of information entropy, we naturally have
\begin{equation}
I(X;Y|Z)=I(Y;X|Z)\leq \min\{H(X),H(Y)\}.
\end{equation}
\end{lemma}

\subsection{Notation}

We follow the information-theoretic approach introduced in \cite{Re:Shah.N.B} and accordingly treat all data symbols including data stored at the storage nodes and those transferred by helper nodes during the repair operations as random variables.

\begin{no}
Throughout the paper, we mainly consider the situation of MSR code with parameter set $\{n=d+1,k,d,\alpha,\beta\}$, because any upper bound on the data file that can be securely stored for any secure MSR code with $\{n=d+1,k,d,\alpha,\beta\}$ also holds for the corresponding secure MSR code with $\{n>d+1,k,d,\alpha,\beta\}$. In Section 3.3, we will establish the equivalent condition of secrecy capacity between any MSR code with $\{n>d+1,k,d,\alpha,\beta\}$ and its truncated one with $\{n=d+1,k,d,\alpha,\beta\}$.
\end{no}

We represent nodes using indices $1$ to $n$ and denote the sequence of nodes $[i,i+1,\cdots,j]$ by $[i,j]$, where $i<j$. We use symbols for a set $\{\ldots\}$ and a sequence $[\ldots]$ interchangeably. For any regenerating code with parameter set $\{n=d+1,k,d,\alpha,\beta\}$, we let

$\blacksquare$ (1). $W_i, i\in[1,d+1]$ denote the random variable corresponding to the content of node $i$. As proved in \cite{Re:Shah.N.B}, it must be that $H(W_i)=\alpha$ for any \emph{optimal} regenerating code including MSR codes.

$\blacksquare$ (2). $\{W_A, A\subseteq [1,d+1]\}$ denote the set of random variables corresponding to the nodes in the subset $A$. Throughout the paper, subscripts of $W$ can represent either a node index or a set of nodes which will be clear from the context.

$\blacksquare$ (3). $S_i^j,\{i,j\}\in[1,d+1],i\neq j$ denote the random variable corresponding to the data symbols sent by the helper node $i$ to the replacement of the failed node $j$. It must be that $H(S_i^j)=\beta$ for any \emph{optimal} regenerating code including MSR codes, following from \cite{Re:Shah.N.B}.

$\blacksquare$ (4). $S_A^B$ denote the set $\{S_i^j|i\in A,j\in B,i\neq j,A\subseteq [1,d+1],B\subseteq [1,d+1]\}$, and particularly $S^B$ substitutes for $S_{[1,d+1]}^B$.

According to the above notation, reconstruction as well as regeneration property of any regenerating code can be expressed as
\begin{equation}
\left\{\begin{aligned}
&H(W_{i_1},W_{i_2},\cdots,W_{i_k})=k\alpha, \quad i_j\in[1,d+1], j\in \{1,\ldots,k\}\\
&H(W_i|S_{\{[1,d+1]\setminus i\}}^i)=0, \quad i\in[1,d+1]
\end{aligned}\right.
\end{equation}

\subsection{Permutation Polynomials}

As shown in \cite{Re:A.Dimakis}, $\alpha$ represents the number of symbols stored in each node and $\beta$ corresponds to the number of symbols downloaded from a surviving node to repair a failed node. Note that the entropy of each symbol cannot be greater than 1 and may not be an integer. Thus, it can only be that $H(W_i)\leq \alpha$ and $H(S_i^j)\leq \beta$. Subsequently, under the context of \emph{optimal} regenerating codes, Shah et al in \cite{Re:Shah.N.B} employ information theory to derive that $H(W_i)=\alpha$ and $H(S_i^j)=\beta$, which implies that each symbol contained in each node and repair data actually reaches the maximum entropy $1$, i.e., each symbol is uniformly distributed inside itself. Besides, it also means that the symbols included in the same node $i$ and same repair data $S_i^j$ are mutually independent respectively. Although each symbol included in any repair data $S_i^j$ has the uniform distribution and $S_i^j$ also has the maximal entropy $\beta$, the joint entropy $H(S_i^{j_1},S_i^{j_2})$ may not be an integer where $j_1\neq j_2$ as illustrated in the following.

We let $(y_1^i,y_2^i,\cdots,y_{\alpha}^i)$ denote the $\alpha$ symbols stored in node $i$, where $H(y_l^i)=1$ for any $l\in[1,\alpha]$. In addition, we let $(z_1^{(i,j_1)},z_2^{(i,j_1)},\cdots,z_{\beta}^{(i,j_1)})$ and $(z_1^{(i,j_2)},z_2^{(i,j_2)},\cdots,z_{\beta}^{(i,j_2)})$ be the $\beta$ symbols contained in the repair data $S_i^{j_1}$ and $S_i^{j_2}$ respectively, where $H(z_l^{(i,j_1)})=H(z_l^{(i,j_2)})=1$ for $l\in [1,\beta]$. Now consider the joint entropy
\begin{equation}
H(S_i^{j_1},S_i^{j_2})=H\big(z_1^{(i,j_1)},z_2^{(i,j_1)},\cdots,z_{\beta}^{(i,j_1)},z_1^{(i,j_2)},z_2^{(i,j_2)},\cdots,z_{\beta}^{(i,j_2)}\big). \end{equation}

In a finite field $\mathbb{F}_q$, any mapping $\tau: \mathbb{F}_q^{n}\rightarrow\mathbb{F}_q$ can be represented by a polynomial over $\mathbb{F}_q$ of degree $<q$ in each ``indeterminate" through Lagrange Interpolation \cite{Re:R. Lidl}. Since all the symbols contained in node $i$ are uniformly distributed inside themselves and mutually independent, they can be regarded as ``indeterminates". So, we let
\begin{equation}
\left\{\begin{aligned}
&z_{1}^{(i,j_1)}=\dot{f}_1(y_1^i,y_2^i,\cdots,y_{\alpha}^i)\\
&z_{2}^{(i,j_1)}=\dot{f}_2(y_1^i,y_2^i,\cdots,y_{\alpha}^i)\\
&\quad \quad\quad\vdots\\
&z_{\beta}^{(i,j_1)}=\dot{f}_{\beta}(y_1^i,y_2^i,\cdots,y_{\alpha}^i)
\end{aligned}\right.
\end{equation}
and
\begin{equation}
\left\{\begin{aligned}
&z_{1}^{(i,j_2)}=\ddot{f}_1(y_1^i,y_2^i,\cdots,y_{\alpha}^i)\\
&z_{2}^{(i,j_2)}=\ddot{f}_2(y_1^i,y_2^i,\cdots,y_{\alpha}^i)\\
&\quad \quad\quad\vdots\\
&z_{\beta}^{(i,j_2)}=\ddot{f}_{\beta}(y_1^i,y_2^i,\cdots,y_{\alpha}^i),
\end{aligned}\right.
\end{equation}
where $(\dot{f}_1,\dot{f}_2,\cdots,\dot{f}_{\beta})$ and $(\ddot{f}_1,\ddot{f}_2,\cdots,\ddot{f}_{\beta})$ represent the polynomials induced by the symbols contained in repair data $S_i^{j_1}$ and $S_i^{j_2}$ respectively. In \cite{Re:R. Lidl}, there are two special concepts introduced as follows.

\begin{definition}\cite{Re:R. Lidl}\label{permutation1}(Permutation Polynomial):
A polynomial $f\in \mathbb{F}_q[x_1,\cdots,x_n]$ is called a permutation polynomial in $n$ indeterminates over $\mathbb{F}_q$ if the equation $f(x_1,\cdots,x_n)=a$ has $q^{n-1}$ solutions in $\mathbb{F}_q^n$ for each $a\in \mathbb{F}_q$.
\end{definition}

According to Definition \ref{permutation1}, we know that each value $a\in\mathbb{F}_q$ will be taken in the same probability ($\frac{q^{n-1}}{q^{n}}=\frac{1}{q}$) by a permutation polynomial. From this point, permutation polynomial exactly corresponds to uniform distribution in information theory (Definition \ref{entropy}). Due to that $H(z_l^{(i,j_1)})=H(z_l^{(i,j_2)})=1$ for any $l\in [1,\beta]$, we know that $(\dot{f}_1,\dot{f}_2,\cdots,\dot{f}_{\beta},\ddot{f}_1,\ddot{f}_2,\cdots,\ddot{f}_{\beta})$ all are permutation polynomials. Here, it should be noted that permutation polynomials are not necessarily linear polynomials in finite fields while linear polynomials apparently are permutation polynomials.

\begin{definition}\cite{Re:R. Lidl}\label{permutation}(Orthogonal System):
A system of polynomials $f_1,\cdots,f_m\in \mathbb{F}_q[x_1,\cdots,x_n]$ where $1\leq m\leq n$ is said to be orthogonal in $\mathbb{F}_q$, if the system of equations
\begin{equation}
\left\{\begin{aligned}
&f_1(x_1,\cdots,x_n)=a_1\\
&\quad \quad\quad\vdots\\
&f_m(x_1,\cdots,x_n)=a_m
\end{aligned}\right.
\end{equation}
has $q^{n-m}$ solutions in $\mathbb{F}_q^n$ for each $(a_1,\cdots,a_m)\in \mathbb{F}_q^m$.
\end{definition}

According to Definition \ref{permutation} and Definition \ref{joint}, we know that $(\dot{f}_1,\dot{f}_2,\cdots,\dot{f}_{\beta})$ and $(\ddot{f}_1,\ddot{f}_2,\cdots,\ddot{f}_{\beta})$ respectively constitute two orthogonal systems, since $H(S_i^{j_1})=H(S_i^{j_2})=\beta$. Similarly, it follows that $H(S_i^{j_1},S_i^{j_2})=2\beta$ if and only if the $2\beta$ polynomials $(\dot{f}_1,\dot{f}_2,\cdots,\dot{f}_{\beta},\ddot{f}_1,\ddot{f}_2,\cdots,\ddot{f}_{\beta})$ can form an orthogonal system.

However, if there exist two different polynomials $\dot{f}_{l_1}$ and $\ddot{f}_{l_2}$ for some $l_1,l_2\in[1,\beta]$ that cannot form an orthogonal system, the joint entropy of the corresponding symbols $H(z_{l_1}^{(i,j_1)},z_{l_2}^{(i,j_2)})$ will be a non-integer, which may result in that all the symbols of repair data $S_i^{\{j_1,j_2\}}$ also have the non-integer joint entropy. Note that multiple permutation polynomials may not form an orthogonal system while each polynomial in an orthogonal system must be a permutation polynomial.

\vspace{0.2cm}
Nevertheless, in the linear context, the joint entropy of the symbols contained in $S_i^A$ must be an integer, where $i\notin A$ and $A$ is any subset of $[1,d+1]$.

\begin{lemma}\label{integer}
In the scenario of linear optimal regenerating codes, the symbols contained in $S_i^A$ must have integer-value entropy, where $i\in[1,d+1]$, $A\subset [1,d+1]$ and $i\notin A$.
\end{lemma}

\begin{proof}
Assume all the $m=|A|\cdot \beta$ symbols in $S_i^A$ can be represented as
\begin{equation}
\{f_1(x_1,x_2,\cdots,x_{\alpha}),f_2(x_1,x_2,\cdots,x_{\alpha}),\cdots,f_{m}(x_1,x_2,\cdots,x_{\alpha})\},
\end{equation}
where $(x_1,x_2,\cdots,x_{\alpha})$ are the $\alpha$ symbols stored in node $i$ and $f_l$ denotes the linear polynomial for $l\in[1,m]$. Then, we let
\begin{equation}\label{gen}
\left\{\begin{aligned}
&f_1(x_1,x_2,\cdots,x_{\alpha})=a_{11}x_1+a_{12}x_2+\cdots+a_{1\alpha}x_{\alpha}\\
&f_2(x_1,x_2,\cdots,x_{\alpha})=a_{21}x_1+a_{22}x_2+\cdots+a_{2\alpha}x_{\alpha}\\
&\quad \quad\quad\quad\quad\quad\quad\quad\quad\quad\vdots\\
&f_m(x_1,x_2,\cdots,x_{\alpha})=a_{m1}x_1+a_{m2}x_2+\cdots+a_{m\alpha}x_{\alpha},
\end{aligned}\right.
\end{equation}
where all the coefficients are drawn from $\mathbb{F}_q$. Equation (\ref{gen}) can be alternatively expressed as
\begin{equation}
(f_1,f_2,\cdots,f_m)^T=C\cdot (x_1,x_2,\cdots,x_{\alpha})^T,
\end{equation}
where $T$ indicates the transpose operation and $C$ denotes the generator matrix.

Since $(f_1,f_2,\cdots,f_m)$ are linear combinations of a set of uniformly distributed random variables, then they all are uniformly distributed and they are either mutually independent, or some of them are determined by the remaining of them. In fact, the value of $H(S_i^A)$ is equal to the rank of $C$, which we denote by $r(C)$.
\vspace{0.2cm}

1. When $m\leq \alpha$ and $r(C)=m$, the row vectors of $C$ are linearly independent. Then, for each vector value $(b_1,b_2,\cdots,b_m)\in\mathbb{F}_q^m$, equation (\ref{gen}) has $q^{\alpha-m}$ solutions in $\mathbb{F}_q^{\alpha}$. Thus, each vector value $(b_1,b_2,\cdots,b_m)$ will occur in equally probability $\frac{q^{\alpha-m}}{q^{\alpha}}=\frac{1}{q^m}$. So, the polynomials (\ref{gen}) form an orthogonal system. In this case, we can calculate that $H(S_i^A)=m=r(C)$ according to Definition \ref{joint}.

2. When $r(C)<m$, the row vectors of $C$ are not linearly independent, which implies $r(C)$ chosen linearly independent polynomials $f_l$ will determine the values of the remaining $m-r(C)$ polynomials. Although the whole polynomials (\ref{gen}) cannot form the orthogonal system, the $r(C)$ linearly independent polynomials still forms an orthogonal system. Similar to the above case, entropy of these $r(C)$ linearly independent polynomials is equal to $r(C)$. Thereby, we have
\begin{equation}
H(S_i^A)=H(f_1,f_2,\cdots,f_m)=H(f_{l_1},f_{l_2},\cdots,f_{l_{r(C)}})=r(C),
\end{equation}
where $\big\{f_{l_1},f_{l_2},\cdots,f_{l_{r(C)}}\big\}$ are the $r(C)$ chosen linearly independent polynomials.

\vspace{0.2cm}
Hence, both cases indicate that $H(S_i^A)=r(C)$, while $r(C)$ must be an integer since it represents the rank of $C$.

\end{proof}

\begin{rem}
In this lemma, theories of permutation polynomial and orthogonal system are used to demonstrate that the joint entropy
of symbols included in multiple sets of repair data in the nonlinear context may be a non-integer value while their joint entropy has to be an integer in the linear context, which is important for the later discussion on secrecy capacity of linear MSR codes.

Additionally, it is of independent interest that these two theories in finite fields are also applicable to the nonlinear context, because they may be utilized to explore the case of constructing nonlinear \emph{optimal} regenerating codes. That is beyond the scope of this paper though. In this paper, we mainly study the secrecy capacity of linear MSR codes, while some new insights on general MSR codes are also present.
\end{rem}

\subsection{A Universal Upper Bound}
\vspace{0.2cm}

$\blacktriangleright$\emph{\textbf{Eavesdropper Model:}} Let $E$ be a set of $l_1$ nodes which the eavesdropper has access to, and $F$ be another disjoint set of $l_2$ nodes whose repair data can be observed by the eavesdropper. In other words, the eavesdropper is assumed to have the knowledge of $\{W_E,S^F\}$. Furthermore, we assume $l_1+l_2<k$, otherwise the eavesdropper can retrieve all the data message. Due to this eavesdropper model, we set $G$ to be another subset $G\subseteq\big\{[1,d+1]\setminus (E\cup F)\big\}$ of size $(k-l_1-l_2)$. Based on this model, a universal upper bound on the secrecy capacity of any \emph{optimal} regenerating code is given as follows.

\begin{lemma}\label{secure expression}
For any secure optimal regenerating code with $\{n=d+1,k,d,\alpha,\beta\}$, we have
\begin{equation}
\left\{\begin{aligned}
&B^{(s)}\\
&\leq H(W_E,W_F,W_G|W_E,S^F)\\
&=H(W_G|W_E,W_F)-H(S^F|W_E,W_F)\\
&=\sum_{i=l_1+l_2+1}^{k}\min\{\alpha,(d-i+1)\beta\}-H(S^F|W_E,W_F)
\end{aligned}\right.
\end{equation}
\end{lemma}

\begin{proof}

First, in secure regenerating codes \cite{Re:N.B.Shah,Re:S.Goparaju,Re:Rawat.A.S}, the random variables associated with the message can be viewed as
the tuple $(D,R)$, where $D$ corresponds to the actual data file and $R$ corresponds to the randomness added. The secure file size is $B^{(s)}=H(D)$ and the secrecy condition requires that $I(D;W_E,S^F)=0$. Thus, it must be that
\begin{equation}
\left\{\begin{aligned}
&H(D)\\
&=H(D)-I(D;W_E,S^F)\\
&=H(D|W_E,S^F)\\
&\leq H(D,R|W_E,S^F)\\
&=H(W_E,W_F,W_G|W_E,S^F),
\end{aligned}\right.
\end{equation}
where the equation in the last step follows from the reconstruction property.

Second, we have
\begin{equation}
\left\{\begin{aligned}
&H(W_G|W_E,W_F)-H(W_E,W_F,W_G|W_E,S^F)\\
&=H(W_G|W_E,W_F)-H(W_E,W_F,W_G|W_E,W_F,S^F)\\
&=H(W_G|W_E,W_F)-H(W_G|W_E,W_F,S^F)\\
&=I(W_G;S^F|W_E,W_F)\\
&=H(S^F|W_E,W_F)-H(S^F|W_E,W_F,W_G)\\
&=H(S^F|W_E,W_F),
\end{aligned}\right.
\end{equation}
where the regeneration property leads to $H(S^F)=H(S^F,W_F)$ that is used in the first step.

At last, for the \emph{optimal} regenerating codes, it follows from the proof of property 1 in \cite{Re:Shah.N.B} that

\begin{equation}
H(W_G|W_E,W_F)=\sum_{i=l_1+l_2+1}^{k}\min\{\alpha,(d-i+1)\beta\}.
\end{equation}
\end{proof}

\begin{rem}\label{secrecy achieve}
In the context of linear regenerating codes, MRD (Maximum Rank Distance) codes \cite{Re:R. M. Roth} (e.g. Gabidulin code \cite{Re:Gabidulin}) can be used to pre-code the original data of size $\{B=k\alpha\}$, which is required to consist of $\{B-H(W_E,S^F)\}$-sized actual data file $D$ and $H(W_E,S^F)$-sized random data $R$. It should be noted that $H(W_E,S^F)$ is also an integer as derived in Lemma \ref{integer}, because $\{W_E,S^F\}$ are obtained by the linear combinations of the original data message of size $B$. As shown in \cite{Re:Rawat.A.S,Re:Koyluoglu}, this kind of secure code construction always can meet the \emph{secrecy condition} that $I(D;W_E,S^F)=0$. It exactly means the maximal file size that can be securely stored is
\begin{equation}
B^{(s)}=B-H(W_E,S^F)= H(W_E,W_F,W_G|W_E,S^F).
\end{equation}

In the MSR scenario, it is obvious that $H(W_G|W_E,W_F)=H(W_G)=(k-l_1-l_2)\alpha$, following from property 2 in \cite{Re:Shah.N.B}. Thus, we only need to concentrate on the term $H(S^F|W_E,W_F)$ in this paper.
\end{rem}

\section{DATA SECRECY FOR GENERAL MSR CODES}

In this section, we give some general properties of MSR codes and a simple expression of upper bound on secrecy capacity, which will be leveraged throughout this paper. Afterwards, \emph{stable} MSR code as a new concept is introduced, where the stable property will be shown to be closely linked with secrecy capacity.

\subsection{Properties of General MSR Codes}

Here, we proceed to provide some new properties of general MSR codes (including the nonlinear context), which actually stem from the reconstruction and regeneration properties of MSR codes. With these properties, we can further simplify the formulation $H(S^F|W_E,W_F)$ mentioned above.

\begin{lemma}\label{conditional entropy}
In the scenario of MSR codes with parameter set $\{n=d+1,k,d,\alpha,\beta\}$, for any node $i$ with efficient repair, consider two arbitrary subsets $A'$ and $B'$ such that $\{|A'|=k-1,|B'|=d-k+1,A'\cap B'=\emptyset,A'\cup B'=[1,d+1]\setminus i\}$, it must be that
\begin{equation}\label{addition pro}
\left\{ \begin{array}{l}
H(S_{A'\cup B'}^i)=d\beta\\
H(S_{B'}^i|W_i,S_{A'}^i)=0.
\end{array} \right.
\end{equation}
\end{lemma}

\begin{proof}Without loss of generality, we assume $i=1$. The proof is given in two steps as follows.

1. According to the Property 2 in \cite{Re:Shah.N.B}, it is trivial that $I(W_1;W_{A'})=0$ in the MSR scenario, which leads to $H(W_1|S_{A'}^1)=\alpha$.\\

2. We set $B'=(b_1,b_2,\cdots, b_{d-k+1})$. Due to the repair property, it must be that $H(W_1|S_{A'}^1,S_{B'}^1)=0$. Next, some key inequalities are present from Lemma \ref{helpful}:
\begin{equation}\label{inequation}
\left\{ \begin{aligned}
&H(W_1|S_{A'}^1)-H(W_1|S_{A'}^1,S_{b_1}^1)\\
&=I(W_1;S_{b_1}^1|S_{A'}^1)\\
&=H(S_{b_1}^1|S_{A'}^1)-H(S_{b_1}^1|W_1,S_{A'}^1)\\
&\leq \beta;\\
&H(W_1|S_{A'}^1,S_{b_1}^1)-H(W_1|S_{A'}^1,S_{b_1}^1,S_{b_2}^1)\\
&=I(W_1;S_{b_2}^1|S_{A'}^1,S_{b_1}^1)\\
&=H(S_{b_2}^1|S_{A'}^1,S_{b_1}^1)-H(S_{b_2}^1|W_1,S_{A'}^1,S_{b_1}^1)\\
&\leq \beta;\\
&\quad \quad\quad\quad\quad\quad\quad\quad\quad\quad\vdots\\
&H(W_1|S_{A'}^1,S_{b_1}^1,S_{b_2}^1,\cdots,S_{b_{d-k}}^1)-H(W_1|S_{A'}^1,S_{B'}^1)\\
&=I(W_1;S_{b_{d-k+1}}^1|S_{A'}^1,S_{\{B'\setminus{b_{d-k+1}}\}}^1)\\
&=H(S_{b_{d-k+1}}^1|S_{A'}^1,S_{\{B'\setminus{b_{d-k+1}}\}}^1)-H(S_{b_{d-k+1}}^1|W_1,S_{A'}^1,S_{\{B'\setminus{b_{d-k+1}}\}}^1)\\
&\leq \beta.
\end{aligned} \right.
\end{equation}
By summing up the left side of the inequalities, we derive
\begin{equation}
\alpha=H(W_1|S_{A'}^1)-H(W_1|S_{A'}^1,S_{B'}^1)\leq (d-k+1)\beta.
\end{equation}
Because $\alpha=(d-k+1)\beta$, it is mandatory that all the inequalities (\ref{inequation}) actually are equations. Thus, for any $j\in[1,d-k+1]$, we have
\begin{equation}\label{addition pro1}
\left\{ \begin{array}{l}
H(S_{b_j}^1|S_{A'}^1,S_{\{b_1,\cdots,b_{j-1}\}}^1)=\beta\\
H(S_{b_j}^1|W_1,S_{A'}^1,S_{\{b_1,\cdots,b_{j-1}\}}^1)=0,
\end{array} \right.
\end{equation}
from which we can derive
\begin{equation}
\left\{\begin{aligned}
&H(S_{A'\cup B'}^1)\\
&=H(S_{A'}^1)+H(S_{B'}^1|S_{A'}^1)\\
&=(k-1)\beta+\sum_{j=1}^{j=d-k+1}H(S_{b_j}^1|S_{A'}^1,S_{\{b_1,\cdots,b_{j-1}\}}^1)\\
&=(k-1)\beta+(d-k+1)\beta\\
&=d\beta
\end{aligned}\right.
\end{equation}
and
\begin{equation}
\left\{\begin{aligned}
&H(S_{B'}^1|W_1,S_{A'}^1)\\
&=\sum_{j=1}^{j=d-k+1}H(S_{b_j}^1|W_1,S_{A'}^1,S_{\{b_1,\cdots,b_{j-1}\}}^1)\\
&=0
\end{aligned}\right.
\end{equation}

\end{proof}

\begin{rem}
This lemma exhibits the special properties of MSR codes. $H(S_{A'\cup B'}^i)=d\beta$ means any repair data from disjoint sets of nodes upon failure of node $i$ are mutually independent. $H(S_{B'}^i|W_i,S_{A'}^i)=0$ implies that given the contents of node $i$ and the repair data from any $k-1$ nodes, the repair data from the remaining $d-k+1$ nodes are deterministic.
\end{rem}

\vspace{0.2cm}

\begin{lemma}\label{Secure size}
In the MSR scenario with $\{n=d+1,k,d,\alpha,\beta\}$, we have
\begin{equation}\label{secure size}
H(S^F|W_E,W_F)=H(S_G^F),
\end{equation}
where $E$, $F$ and $G$ are pairwise disjoint sets as defined in Section 2.4 and $|E\cup F\cup G|=k$. Furthermore, when $E=\emptyset $, we still have $H(S^F|W_F)=H(S_G^F)$, where $|F\cup G|=k$.
\end{lemma}

\begin{proof}

Assume all the $d+1$ nodes are comprised of $E,F,G$ and $T$, where $|E\cup F\cup G|=k$ and $|T|=d-k+1$. Thereby, we have
\begin{equation}
\left\{\begin{aligned}
&H(S^F|W_{\{E,F\}})\\
&=H(S_{\{E,F,G,T\}}^F|W_{\{E,F\}})\\
&=H(S_{\{G,T\}}^F|W_{\{E,F\}})\\
&=H(S_G^F|W_{\{E,F\}})+H(S_T^F|W_{\{E,F\}},S_G^F).
\end{aligned}\right.
\end{equation}

Then, due to the condition $|(E,F,G)\setminus i|=k-1$, Lemma \ref{conditional entropy} leads to that for any $i\in F$,
\begin{equation}
\left\{\begin{aligned}
&H(S_T^i|W_{\{E,F\}},S_G^F)\\
&\leq H(S_T^i|W_{\{E,F\}},S_G^i)\\
&=H(S_T^i|W_i,W_{\{(E,F)\setminus i\}},S_G^i)\\
&\leq H(S_T^i|W_i,S_{\{(E,F)\setminus i\}}^i,S_G^i)\\
&=H(S_T^i|W_i,S_{\{(E,F,G)\setminus i\}}^i)\\
&=0,
\end{aligned}\right.
\end{equation}
from which we derive $H(S_T^F|W_{\{E,F\}},S_G^F)=0$.

Still by the Property 2 in \cite{Re:Shah.N.B}, it has to be that $H(S_G^F|W_E,W_F)=H(S_G^F)$, since $|E\cup F\cup G|=k$. Hence, we obtain the proof. In addition, the above deduction is obviously applicable to the situation when $E=\emptyset$.
\end{proof}

\begin{rem}\label{simply}
Combining this lemma with Lemma \ref{secure expression}, we have
\begin{equation}
B^{(s)}\leq(k-l_1-l_2)\alpha-H(S_G^F).
\end{equation}
In particular, for the linear MSR codes, this upper bound can always be achieved (as in Remark \ref{secrecy achieve}). Moreover, the equation $H(S^F|W_F)=H(S_G^F)$ promotes the next result.
\end{rem}

\vspace{0.2cm}

\begin{lemma}\label{interesting1}
In the MSR scenario with $\{n=d+1,k,d,\alpha,\beta\}$, for any subset $F$ such that $|F|\leq k-1$, and arbitrary different $i_1,i_2$ where $i_1,i_2\notin F$, we have $H(S_{i_1}^F)=H(S_{i_2}^F)$.
\end{lemma}

\begin{proof}
According to Lemma \ref{Secure size}, we obtain
\begin{equation}\label{general eq}
\left\{\begin{aligned}
&H(S^F)\\
&=H(S^F,W_F)\\
&=H(W_F)+H(S^F|W_F)\\
&=H(W_F)+H(S_{G'}^F|W_F)\\
&=H(W_F)+H(S_{G'}^F),
\end{aligned}\right.
\end{equation}
where $G'$ is any subset of $[1,d+1]$ such that $|G'|+|F|=k$ and $G' \cap F={\O}$.
Based on the condition $|F|\leq k-1$, then it has to be that $|G'|\geq 1$.

\vspace{0.2cm}

1. When $|G'|=1$, for any two different $g_1$ and $g_2$ where $g_1,g_2\in \{[1,d+1]\setminus F\}$,
\begin{equation}
H(S^F)=H(W_F)+H(S_{g_1}^F)=H(W_F)+H(S_{g_2}^F),
\end{equation}
which indicates $H(S_{g_1}^F)=H(S_{g_2}^F)$.

2. When $|G'|\geq 2$, we set $G'=\{g',G_1\}$ and $G''=\{g'',G_1\}$ such that $\{g'\neq g'',|G'|=|G''|=k-|F|,G'\cap F=G''\cap F=\emptyset\}$, where $G''$ plays the same role of $G'$ in the following statement. Similarly, we derive
\begin{equation}
\left\{\begin{aligned}
&H(S^F)\\
&=H(W_F)+H(S_{G'}^F)\\
&=H(W_F)+H(S_{g'}^F)+H(S_{G_1}^F);\\
&H(S^F)\\
&=H(W_F)+H(S_{G''}^F)\\
&=H(W_F)+H(S_{g''}^F)+H(S_{G_1}^F),
\end{aligned}\right.
\end{equation}
which implies $H(S_{g'}^F)=H(S_{g''}^F)$.

Because the choice of $(g_1,g_2)$ and $(g',g'')$ are arbitrary, then for arbitrary different $i_1,i_2$ where $i_1,i_2\notin F$, we have $H(S_{i_1}^F)=H(S_{i_2}^F)$.

\end{proof}

\begin{rem}\label{iden}
Lemma \ref{interesting1} shows that the entropy of repair data from any two nodes assisting in repairing the same subsets of nodes are identical. Combining this lemma with Remark \ref{simply}, we further obtain the following result on upper bound.
\end{rem}

\subsection{A Simple Expression of Upper Bound}

Incorporating Lemma \ref{secure expression}, Lemma \ref{Secure size} and Lemma \ref{interesting1}, we consequently derive a simple and generally applicable result on secrecy capacity as follows.

\begin{theorem}\label{sim}
In the scenario of MSR codes with parameter set $\{n=d+1,k,d,\alpha,\beta\}$, we have
\begin{equation}\label{secure size further}
B^{(s)}\leq (k-l_1-l_2)\big(\alpha-H(S_g^F)\big),
\end{equation}
where $g\in G$, $|G|=k-l_1-l_2$, and $|F|=l_2$.
\end{theorem}

\begin{rem}\label{simple}
This can be viewed as a simple and generally applicable upper bound of $B^{(s)}$, since we only need to calculate or estimate the joint entropy of repair data transmitted from any single node $H(S_g^F)$. Still by Remark \ref{secrecy achieve}, this upper bound can be reached in the scenario of linear MSR codes.
\end{rem}

\subsection{Stable MSR Codes}

Given an MSR code with $\{n=d+1,k,d,\alpha,\beta\}$, since our focus is the exact repair, the random variables $W_j$ are invariant with time, i.e., they remain constant irrespective of the sequence of failures and repairs that occur in the storage system. Once construction of such an MSR code with $\{n=d+1\}$ is present, content of $S_i^j$ sent from a node $i$ to repair another node $j$ also keeps invariant. However, for the MSR code with $\{n>d+1,k,d,\alpha,\beta\}$, the repair data $S_i^j$ technically need not keep constant and may vary with different sets of helper nodes including the same node $i$, only needing to satisfy that per node storage $W_j$ stays unchanged. For instance, when node $j$ is failed, node $i$ is assigned to assist in repairing node $j$. Thus, there totally exists ${d-1}\choose{n-2}$ possible sets of helper nodes including node $i$.

Assume repair data of node $j$ is captured by the eavesdropper\footnote{It would be reasonable to assume here that
the identity of node $j$ can be recognized by eavesdropper, although node $j$ when failed will be replaced by newcomer nodes. In this case, the eavesdropper will gain access to all repair data via sitting on the same node $j$ undergoing different repair epochs.}. If content of repair data $S_i^j$ is not independent of the choice of the set of helper nodes and varies with them, after multiple repair epochs with different sets of helper nodes including node $i$, different information regarding repair data $S_i^j$ will be exposed to the eavesdropper. Thus, the eavesdropper is supposed to observe more information regarding repair data of node $j$, when compared to the case of invariant content of repair data. In the following, we will use an example to illustrate this security issue.

\begin{exa}\label{unstable}
Assume $E=\emptyset$ and $F=\{1\}$, i.e., only the repair data of node $1$ is eavesdropped. Consider two truncated MSR codes $\mathbb{M}$ and $\mathbb{M'}$ comprised of nodes set $[1,d+1]$ and $[1,3,\cdots,d+2]$ respectively from an MSR code with $\{n>d+1,k,d,\alpha,\beta\}$.

Since they still are MSR codes, they necessarily retain the properties in Section 3.1. Thus, we have
\begin{equation}
\left\{\begin{aligned}
&H\big(S^1(\mathbb{M})\big)=H(W_1)+(k-1)H\big(S_3^1(\mathbb{M})\big)=\alpha+(k-1)\beta=d\beta\\
&H\big(S^1(\mathbb{M'})\big)=H(W_1)+(k-1)H\big(S_3^1(\mathbb{M'})\big)=\alpha+(k-1)\beta=d\beta,
\end{aligned}\right.
\end{equation}
where $H\big(S^1(\mathbb{M})\big)$ and $H\big(S^1(\mathbb{M'})\big)$ respectively represent the repair data of node $1$ under different contexts of truncated MSR codes. Besides, it follows from Lemma \ref{interesting1} that $H\big(S_{i_1}^1(\mathbb{M})\big)=H\big(S_3^1(\mathbb{M})\big)$ and $H\big(S_{i_2}^1(\mathbb{M'})\big)=H\big(S_3^1(\mathbb{M'})\big)$ for any $i_1\in\mathbb{M}$ and $i_2\in\mathbb{M'}$.

Furthermore, similar to the deduction of properties given in Section 3.1, we derive
\begin{equation}
\left\{\begin{aligned}
&H\big(S^1(\mathbb{M}),S^1(\mathbb{M'})\big)\\
&=H\big(W_1,S^1(\mathbb{M}),S^1(\mathbb{M'})\big)\\
&=H(W_1)+H\big(S^1(\mathbb{M}),S^1(\mathbb{M'})|W_1\big)\\
&=H(W_1)+H\big(S_{[2,d+1]}^1(\mathbb{M}),S_{[3,d+2]}^1(\mathbb{M'})|W_1\big)\\
&=H(W_1)+H\big(S_{[3,k+1]}^1(\mathbb{M},\mathbb{M'})|W_1\big)
+H\big(S_{[2,k+2,\cdots,d+1]}^1(\mathbb{M}),S_{[k+2,d+2]}^1(\mathbb{M'})|W_1,S_{[3,k+1]}^1(\mathbb{M},\mathbb{M'})\big)\\
&=H(W_1)+H\big(S_{[3,k+1]}^1(\mathbb{M},\mathbb{M'})|W_1\big)\\
&=H(W_1)+H\big(S_{[3,k+1]}^1(\mathbb{M},\mathbb{M'})\big)\\
&=H(W_1)+(k-1)H\big(S_3^1(\mathbb{M},\mathbb{M'})\big),
\end{aligned}\right.
\end{equation}
where $H\big(S_{[2,k+2,\cdots,d+1]}^1(\mathbb{M}),S_{[k+2,d+2]}^1(\mathbb{M'})|W_1,S_{[3,k+1]}^1(\mathbb{M},\mathbb{M'})\big)=0$ results from Lemma \ref{conditional entropy}.

If $S_3^1(\mathbb{M})$ does not share the same information with $S_3^1(\mathbb{M'})$, it has to be that $H\big(S_3^1(\mathbb{M},\mathbb{M'})\big)>\beta$, which leads to that $H\big(S^1(\mathbb{M},\mathbb{M'})\big)>d\beta$.
It means that eavesdropper will inevitably obtain different data information after multiple repair epochs with different sets of helper nodes. When traversing all possible truncated MSR codes corresponding to repair epochs with all possible sets of helper nodes, it even may render the storage system unable to maintain any data secrecy.
\end{exa}

Based on the above security concern, we define a special MSR code as follows.

\begin{definition}(Stable MSR Code): A stable MSR code with $\{n>d+1,k,d,\alpha,\beta\}$ is an MSR code with the ``stable" repair property, i.e., the data transmitted from any node $i$ to repair node $j$ is independent of the set of helper nodes including the same node $i$. In other words, content of $S_i^j$ remains invariant under different sets of helper nodes including the same node $i$.
\end{definition}

One can check that the product-matrix-based MSR code \cite{Re:K.Rashmi} is a \emph{stable} MSR code. The following theorem will show that this stable property in fact is the equivalent condition of secrecy capacity between any MSR code with $\{n>d+1,k,d,\alpha,\beta\}$ and its truncated one with $\{n=d+1,k,d,\alpha,\beta\}$.

\begin{lemma}\label{identical}
Let $\mathbb{N}$ be a stable MSR code with the parameter set $\{n>d+1,k,d,\alpha,\beta\}$ and $\mathbb{N'}$ be the stable MSR code with $\{n=d+1,k,d,\alpha,\beta\}$ truncated from $\mathbb{N}$, then the secrecy capacity of $\mathbb{N}$ is same as that of $\mathbb{N'}$.
\end{lemma}

\begin{proof}

Without loss of generality, assume $\mathbb{N'}$ is comprised of the nodes set $[1,d+1]$ truncated from $\mathbb{N}$. We set the same subsets $E,F,G$ for $\mathbb{N}$ and $\mathbb{N'}$, where $E,F,G$ are three disjoint subsets of $[1,d+1]$ as defined in section 2.4.

Lemma \ref{secure expression} indicates, for any secure regenerating code with $\{n=d+1,k,d,\alpha,\beta\}$,
\begin{equation}\label{gb}
B^{(s)}\leq H(W_G|W_E,W_F)-H(S^F|W_E,W_F).
\end{equation}
Although this universal upper bound is established on regenerating codes with length equaling to $\{d+1\}$, it actually is also applicable to those extended regenerating codes with $\{n>d+1\}$ since they still have the reconstruction and regeneration properties. Nevertheless, in order to avoid confusion, we let $S^F(\mathbb{N})$ and $S^F(\mathbb{N'})$ respectively represent the repair data of the nodes set $F$ under the contexts of $\mathbb{N}$ and $\mathbb{N'}$. Accordingly, we have
\begin{equation}
\left\{\begin{aligned}
&B^{(s)}(\mathbb{N})\leq H(W_G|W_E,W_F)-H\big(S^F(\mathbb{N})|W_E,W_F\big) \\
&B^{(s)}(\mathbb{N'})\leq H(W_G|W_E,W_F)-H\big(S^F(\mathbb{N'})|W_E,W_F\big),
\end{aligned}\right.
\end{equation}
with which we only need to prove that $H\big(S^F(\mathbb{N})|W_E,W_F\big)=H\big(S^F(\mathbb{N'})|W_E,W_F\big)$. Since $\mathbb{N}$ and $\mathbb{N'}$ both are \emph{stable} MSR codes, we can unambiguously substitute $S^F(\mathbb{N})=S^F_{[1,n]}$ and $S^F(\mathbb{N'})=S^F_{[1,d+1]}$. Thus, it follows by showing that
\begin{equation}
H(S^F_{[1,n]}|W_E,W_F)=H(S^F_{[1,d+1]}|W_E,W_F),
\end{equation}
with which it is sufficient to prove that $H(S^F_{[1,n]}|S^F_{[1,d+1]})=0$. To this end, it is equivalent to prove that $H(S^i_{[1,n]}|S^i_{[1,d+1]})=H(S^i_{[d+2,n]}|S^i_{[1,d+1]})=0$ for any $i\in F$.

Without any loss of generality, we consider the situation when $i=1$. Thus, we have
\begin{equation}\label{equivalent}
\left\{\begin{aligned}
&H(S^1_{[d+2,n]}|S^1_{[1,d+1]})\\
&=H(S^1_{[d+2,n]}|S^1_{[2,d+1]})\\
&=H(S^1_{[d+2,n]}|W_1,S_{[2,d+1]}^1)\\
&\leq H(S^1_{[d+2,n]}|W_1,S_{[2,k]}^1).
\end{aligned}\right.
\end{equation}

1. When $n-d-1\geq d-k+1$, we set $Q$ is any subset of $[d+2,n]$ of size $d-k+1$. Because $\mathbb{N}$'s any truncated code with $\{d+1,k,d,\alpha,\beta\}$ still is an MSR code, the nodes set $\{[2,k]\cup Q\}$ can be viewed as a truncated MSR code. Due to the second term of equation (\ref{addition pro}) in Lemma \ref{conditional entropy}, we further derive $H(S^1_Q|W_1,S_{[2,k]}^1)=0$. Since $Q$ is a random subset of $[d+2,n]$, it is obvious that $H(S^1_{[d+2,n]}|W_1,S_{[2,k]}^1)=0$.

\vspace{0.2cm}
2. When $n-d-1< d-k+1$, we set $Q$ is any $d-k+1$-sized set such that $[d+2,n]\subset Q$ and $[1,k]\cap Q=\emptyset$. Similarly, we have $H(S^1_Q|W_1,S_{[2,k]}^1)=0$, from which we can also derive $H(S^1_{[d+2,n]}|W_1,S_{[2,k]}^1)=0$.

Combined with formula (\ref{equivalent}), both cases imply that $H(S^1_{[d+2,n]}|S^1_{[1,d+1]})=0$.

\end{proof}

\begin{rem}\label{stable}
Lemma \ref{identical} indicates that secrecy capacity of \emph{stable} MSR codes does not depend on the parameter $n$ but the remaining parameters $\{k,d,\alpha,\beta,B\}$. One example of \emph{stable} MSR codes with $\{n>d+1\}$ is the product-matrix-based MSR code given by Rashmi et al \cite{Re:K.Rashmi}. In another aspect, it is an interesting question to design an MSR code with unstable property. However, for any unstable MSR code with $\{n>d+1\}$, its secrecy capacity is strictly less than that of the corresponding truncated one with $\{n=d+1\}$, as shown in Example \ref{unstable}. Thus, this stable property is highly advantageous in constructing secure MSR codes.

\end{rem}

\begin{no}

In subsequent discussion, we focus on the secrecy capacity of linear MSR codes with $\{n=d+1,k,d,\alpha,\beta\}$.
\end{no}

\section{SECRECY CAPACITY OF LINEAR MSR CODES}

In this section, we will give a comprehensive and explicit result on secrecy capacity for linear MSR codes with $\{n=d+1,k,d,\alpha,\beta\}$, which is divided into two categories. In the first category, the secrecy capacity is fully characterized, which applies to all linear scalar MSR codes, i.e., $\beta=1$. In the second category, upper bounds on secrecy capacity are present, which apply to all known vector codes with $\{\beta=(d-k+1)^x\}$ where $x\geq 1$ such as Zigzag code \cite{Re:I.Tamo}. Furthermore, these two categories will be shown to also apply to those unexplored linear vector MSR codes with $\{1<\beta<d-k+1\}$. Before these, we first give a lemma that will be used in the subsequent proofs.

\begin{lemma}\label{expre}
Given any regenerating code with $\{n=d+1,k,d,\alpha,\beta\}$, for any set $J=(j_1,j_2,\cdots,j_m)\subseteq[1,d+1]$, we have
\begin{equation}\label{express}
H(S^J)=H(S_{\{[1,d+1]\setminus j_1\}}^{j_1}, S_{\{[1,d+1]\setminus (j_1,j_2)\}}^{j_2}, \cdots, S_{\{[1,d+1]\setminus (j_1,j_2,\cdots,j_{m})\}}^{j_m}).
\end{equation}
\end{lemma}

\begin{proof}

The proof can be obtained from two directions.

First, it is clear that
\begin{equation}
\left\{\begin{aligned}
&H(S_{\{[1,d+1]\setminus j_1\}}^{j_1}, S_{\{[1,d+1]\setminus (j_1,j_2)\}}^{j_2}, \cdots, S_{\{[1,d+1]\setminus (j_1,j_2,\cdots,j_{m})\}}^{j_m})\\
&\leq H(S^J).
\end{aligned}\right.
\end{equation}

Second, we can deduce that
\begin{equation}
\left\{\begin{aligned}
&H(\underbrace{S_{\{[1,d+1]\setminus j_1\}}^{j_1}}, S_{\{[1,d+1]\setminus (j_1,j_2)\}}^{j_2}, \cdots, S_{\{[1,d+1]\setminus (j_1,j_2,\cdots,j_{m})\}}^{j_m})\\
&=H(\underbrace{S_{\{[1,d+1]\setminus j_1\}}^{j_1},W_{j_1}}, S_{\{[1,d+1]\setminus (j_1,j_2)\}}^{j_2}, \cdots, S_{\{[1,d+1]\setminus (j_1,j_2,\cdots,j_{m})\}}^{j_m})\\
&=H(\underbrace{S_{\{[1,d+1]\setminus j_1\}}^{j_1},W_{j_1},S_{j_1}^{\{j_2,\cdots,j_m\}}}, S_{\{[1,d+1]\setminus (j_1,j_2)\}}^{j_2}, \cdots, S_{\{[1,d+1]\setminus (j_1,j_2,\cdots,j_{m})\}}^{j_m})\\
&=H(S_{\{[1,d+1]\setminus j_1\}}^{j_1},W_{j_1},S_{\{[1,d+1]\setminus j_2\}}^{j_2}, \cdots, S_{\{[1,d+1]\setminus (j_2,\cdots,j_{m})\}}^{j_m})\\
&\quad \quad\quad\quad\quad\quad\quad\quad\quad\quad\quad\quad\quad\quad\quad\quad\vdots\\
&=H(S_{\{[1,d+1]\setminus j_1\}}^{j_1},W_{j_1},S_{\{[1,d+1]\setminus j_2\}}^{j_2},W_{j_2}, \cdots, S_{\{[1,d+1]\setminus j_{m}\}}^{j_m},W_{j_m})\\
&\geq H(S_{\{[1,d+1]\setminus j_1\}}^{j_1},S_{\{[1,d+1]\setminus j_2\}}^{j_2}, \cdots, S_{\{[1,d+1]\setminus j_{m}\}}^{j_m})\\
&=H(S^J),
\end{aligned}\right.
\end{equation}
where the formulas in the braces follow from that, for any $l\in [1,m]$,
\begin{equation}
\left\{\begin{aligned}
&H(S_{\{[1,d+1]\setminus j_l\}}^{j_l})=H(S_{\{[1,d+1]\setminus j_l\}}^{j_l},W_{j_l})\\
&H(S_{j_l}^{\{j_{l+1},\cdots,j_m\}},W_{j_l})=H(W_{j_l}).
\end{aligned}\right.
\end{equation}

\end{proof}

\begin{rem}
This lemma indicates that there exist much dependence among repair data of multiple sets of nodes. With it, we can reduce the amount of helper nodes for some failed nodes. Thereby, we can derivatively obtain
\begin{equation}
H(S_{\{[1,d+1]\setminus (j_1,j_2,\cdots,j_{m})\}}^{j_m}|S^{\{j_1,j_2,\cdots,j_{m-1}\}})=
H(S^{\{j_1,j_2,\cdots,j_{m}\}})-H(S^{\{j_1,j_2,\cdots,j_{m-1}\}}),
\end{equation}
which will be used in the proofs later.
\end{rem}

\vspace{0.1cm}

\subsection{Category 1: Precise Value of Secrecy Capacity}
Here, we will give the precise value of secrecy capacity for linear MSR codes with $\{1\leq\beta<\frac{d-k+1}{l_2-1}\}$ and as a result prove the optimality of the secure product-matrix-based MSR codes given in \cite{Re:N.B.Shah}.

\subsubsection{4.1.1 The situation when $\beta=1$.}

\begin{theorem}\label{third proof}
In the linear MSR scenario, for any subsets $P$ and $T$ with $\{|P|=k,|T|=d-k+1,P\cap T=\emptyset\}$, any $F$ such that $F\subseteq T$ and $|F|\leq k-1$, and arbitrary $i\notin F$, we have $H(S_{i}^F)=|F|\beta=|F|$ when $\beta=1$.
\end{theorem}

\begin{proof}
Assume $P=[1,k]$ and $T=[k+1,d+1]$. Without loss of generality, also assume $F=[k+1,k+c]$, where $c\geq 2$ (as it is trivial when $c=1$). In the linear MSR scenario, Lemma $\ref{integer}$ indicates that $H(S_i^A)$ has to be an integer for any subset $A\subseteq F$ and $i\notin A$.

By proof of contradiction, under the condition $\beta=1$, we assume $c$ is the smallest value satisfying that $H(S_1^{[k+1,k+c-1]})=H(S_1^{[k+1,k+c]})=(c-1)\beta$. Based on Lemma \ref{interesting1}, we know that for any $i\notin [k+1,k+c]$, it must be that $H(S_i^{[k+1,k+c-1]})=H(S_i^{[k+1,k+c]})=(c-1)\beta$, from which we further derive $H(S_i^{k+c}|S_i^{[k+1,k+c-1]})=0$. Then, following from Lemma \ref{expre}, we have that for any $j\in[1,d-k+1]$,
\begin{equation}\label{express}
\left\{\begin{aligned}
&H(S^{[k+1,k+j]})=H(S_{[1,k]\cup [k+2,d+1]}^{k+1}, S_{[1,k]\cup [k+3,d+1]}^{k+2}, \cdots, S_{[1,k]\cup [k+j+1,d+1]}^{k+j})\\
&H(S_{[1,k]\cup [k+j+1,d+1]}^{k+j}|S^{[k+1,k+j-1]})=H(S^{[k+1,k+j]})-H(S^{[k+1,k+j-1]})
\end{aligned}\right.
\end{equation}

\vspace{0.1cm}

In one way, since we have $H(S_i^{k+c}|S_i^{[k+1,k+c-1]})=0$ for any $i\notin [k+1,k+c]$ from the above assumption, we have
\begin{equation}\label{easy1}
\left\{\begin{aligned}
&H(S_{[1,k]\cup [k+c+1,d+1]}^{k+c}|S^{[k+1,k+c-1]})\\
&=H(S_{[1,k]\cup [k+c+1,d+1]}^{k+c}|S_{[1,k]\cup [k+2,d+1]}^{k+1}, S_{[1,k]\cup [k+3,d+1]}^{k+2}, \cdots, S_{[1,k]\cup [k+c,d+1]}^{k+c-1})\\
&\leq H(S_{[1,k]\cup [k+c+1,d+1]}^{k+c}|S_{[1,k]\cup [k+c+1,d+1]}^{[k+1,k+c-1]})\\
&=0.
\end{aligned}\right.
\end{equation}

In another way, we derive
\begin{equation}\label{easy2}
\left\{\begin{aligned}
&H(S_{[1,k]\cup [k+c+1,d+1]}^{k+c}|S^{[k+1,k+c-1]})\\
&=H(S^{[k+1,k+c]})-H(S^{[k+1,k+c-1]})\\
&=\{H(W_{[k+1,k+c]})+H(S_{G'}^{[k+1,k+c]})\}-\{H(W_{[k+1,k+c-1]})+H(S_{G''}^{[k+1,k+c-1]})\}\\
&=\{c\alpha+(k-c)(c-1)\beta\}-\{(c-1)\alpha+(k-c+1)(c-1)\beta\}\\
&=\alpha-(c-1)\beta\\
&=(d-k-c+2)\beta,
\end{aligned}\right.
\end{equation}
where
\begin{equation}
\left\{\begin{aligned}
&H(S^{[k+1,k+c]})=H(W_{[k+1,k+c]})+H(S_{G'}^{[k+1,k+c]})\\
&H(S^{[k+1,k+c-1]})=H(W_{[k+1,k+c-1]})+H(S_{G''}^{[k+1,k+c-1]})
\end{aligned}\right.
\end{equation}
result from Lemma \ref{Secure size} and $G'$, $G''$ are defined as in Lemma \ref{interesting1} with $|G'|=k-c$ and $|G''|=k-c+1$.

Now, we are to make comparison between equation (\ref{easy1}) and (\ref{easy2}), when $c\leq \min\{d-k+1,k-1\}$. Equation (\ref{easy2}) is a monotone decreasing function in the variable $c$, thus there are two cases as follows.

\vspace{0.2cm}

1. If $d-k+1\geq k-1$, when $c=k-1$, equation (\ref{easy2}) takes minimum value $\{d-2k+3\}\beta$ that is strictly greater than $0$.

\vspace{0.2cm}

2. If $d-k+1\leq k-1$, when $c=d-k+1$, equation (\ref{easy2}) reaches minimum value $\beta$ that is still positive.
\vspace{0.1cm}

To this end, both cases indicate that equation (\ref{easy2}) contradicts formula (\ref{easy1}), when $c\leq \min\{d-k+1,k-1\}$, i.e., the assumption that $H(S_i^{[k+1,k+c-1]})=H(S_i^{[k+1,k+c]})$ cannot hold. In other words, there does not exist such value $c$ that $H(S_i^{[k+1,k+c-1]})=H(S_i^{[k+1,k+c]})$, when $\beta=1$ and $c\leq \min\{d-k+1,k-1\}$. Therefore, we can claim that, for any $F$ such that $F\subseteq T$ and $|F|\leq k-1$, $H(S_i^F)=H(S_i^{[k+1,k+c]})=c\beta$.

\end{proof}

\vspace{0.2cm}

\begin{corollary}\label{Secure size further1}
In the linear MSR scenario, when $\beta=1$, we have
\begin{equation}\label{secure size further1}
B^{(s)}=(k-l_1-l_2)(\alpha-l_2\beta),
\end{equation}
where $l_1+l_2\leq k-1$ and $l_2\leq d-k+1$.
\end{corollary}

\begin{proof}
Remark \ref{simple} implies that $B^{(s)}=(k-l_1-l_2)\big(\alpha-H(S_g^F)\big)$ in the linear MSR scenario, where $g\notin F$. Combining it with Theorem \ref{third proof}, we obtain this corollary.
\end{proof}

\begin{corollary}\label{last}
The product-matrix-based secure MSR code given in \cite{Re:N.B.Shah} is optimal for any $l_1+l_2\leq k-1$ and $l_2\leq d-k+1$.
\end{corollary}

\begin{proof}
First, the product-matrix-based MSR codes constructed in \cite{Re:K.Rashmi} is established on $\beta=1$ and is a
\emph{stable} MSR code as stated in Remark \ref{stable}. Then, according to the construction of secure MSR codes in \cite{Re:N.B.Shah}, the $(l_1,l_2)$-secure MSR code achieves
\begin{equation}
B^{(s)}=(k-l_1-l_2)(\alpha-l_2\beta).
\end{equation}
Thus, the secrecy capacity of secure MSR codes in \cite{Re:N.B.Shah} exactly complies with that given in Corollary \ref{Secure size further1}.
\end{proof}

\begin{rem}
Actually, Corollary \ref{Secure size further1} is applicable to all linear scalar MSR codes, i.e., linear MSR codes with $\beta=1$. In other words, by MRD code's pre-coding as stated in Remark \ref{secrecy achieve}, all linear scalar secure MSR codes can offer this secrecy capacity with precise value given in Corollary \ref{Secure size further1}.
\end{rem}

\subsubsection{4.1.2 The situations when $1\leq\beta<\frac{d-k+1}{l_2-1}$.}

\begin{theorem}\label{fourth proof}
In the linear MSR scenario, for any subsets $P$ and $T$ where $\{|P|=k,|T|=d-k+1,P\cap T=\emptyset\}$, any $F$ such that $F\subseteq T$ and $|F|\leq k-1$, and arbitrary $i\notin F$, when $\beta<\frac{d-k+1}{|F|-1}$ or $|F|<1+\frac{d-k+1}{\beta}$, we have $H(S_{i}^F)=|F|\beta$ where $\beta>1$.
\end{theorem}

\begin{proof}
Similar to Theorem \ref{third proof}, we assume $P=[1,k],T=[k+1,d+1]$ and $F=[k+1,k+c]$ where $c\geq 2$.

By proof of contradiction, we assume $c$ is the smallest value such that $H(S_1^{[k+1,k+c-1]})=(c-1)\beta$ and $H(S_1^{[k+1,k+c]})=(c-1)\beta+\theta$, where $\theta\in[0,\beta-1]$ and $\theta$ must be an integer following from Lemma \ref{integer}, when $\beta>1$. From Lemma \ref{interesting1}, we know $H(S_i^{[k+1,k+c-1]})=(c-1)\beta$ and $H(S_i^{[k+1,k+c]})=(c-1)\beta+\theta$, where $\theta\in \mathbb{Z}\cap[0,\beta-1]$ for any $i\notin [k+1,k+c]$, from which we further have $H(S_i^{k+c}|S_i^{[k+1,k+c-1]})=\theta$.

\vspace{0.2cm}

Due to the similar way of Lemma \ref{expre} used in the proof of Theorem \ref{third proof}, we first have
\begin{equation}\label{easy3}
\left\{\begin{aligned}
&H(S_{[1,k]\cup [k+c+1,d+1]}^{k+c}|S^{[k+1,k+c-1]})\\
&=H(S^{[k+1,k+c]})-H(S^{[k+1,k+c-1]})\\
&=\{H(W_{[k+1,k+c]})+H(S_{G'}^{[k+1,k+c]})\}-\{H(W_{[k+1,k+c-1]})+H(S_{G''}^{[k+1,k+c-1]})\}\\
&=\{c\alpha+(k-c)[(c-1)\beta+\theta]\}-\{(c-1)\alpha+(k-c+1)(c-1)\beta\}\\
&=(d-k-c+2)\beta+(k-c)\theta.
\end{aligned}\right.
\end{equation}

In another way, we obtain
\begin{equation}\label{easy4}
\left\{\begin{aligned}
&H(S_{[1,k]\cup [k+c+1,d+1]}^{k+c}|S^{[k+1,k+c-1]})\\
&=H(S_1^{k+c}|S^{[k+1,k+c-1]})+H(S_2^{k+c}|S^{[k+1,k+c-1]}, S_1^{k+c})+\cdots+H(S_{d+1}^{k+c}|S^{[k+1,k+c-1]},S_{[1,k]\cup [k+c+1,d]}^{k+c})\\
&\leq H(S_1^{k+c}|S_1^{[k+1,k+c-1]})+H(S_2^{k+c}|S_2^{[k+1,k+c-1]})+\cdots+H(S_{d+1}^{k+c}|S_{d+1}^{[k+1,k+c-1]})\\
&=(d-c+1)\theta.
\end{aligned}\right.
\end{equation}

Thus, if $(d-c+1)\theta<(d-k-c+2)\beta+(k-c)\theta$, i.e., $(d+1-k)\theta<(d-k-c+2)\beta$, contradiction arises. Particularly, when $\theta=\beta-1$, $(d+1-k)\theta$ reaches maximum $(d+1-k)(\beta-1)$. By simplification, we obtain that, when $\beta<\frac{d-k+1}{c-1}$ or $c<1+\frac{d-k+1}{\beta}$, equations (\ref{easy3}) contradicts formula (\ref{easy4}), which means the assumption that $H(S_1^{[k+1,k+c-1]})=(c-1)\beta$ and $H(S_1^{[k+1,k+c]})=(c-1)\beta+\theta$ cannot hold, when $\theta\in \mathbb{Z}\cap[0,\beta-1]$. That is to say, the value of $\theta$ here can only be exactly taken by $\beta$.

Therefore, for any $F$ such that $F\subseteq T$ and $|F|\leq k-1$, when $\beta<\frac{d-k+1}{c-1}$ or $c<1+\frac{d-k+1}{\beta}$, we have $H(S_i^F)=H(S_i^{[k+1,k+c]})=c\beta$.

\end{proof}

\begin{corollary}\label{Secure size further2}
In the linear MSR scenario, when $\beta\geq 1$, we still have
\begin{equation}\label{secure size further2}
B^{(s)}=(k-l_1-l_2)(\alpha-l_2\beta),
\end{equation}
when $l_1+l_2\leq k-1$ and $l_2<1+\frac{d-k+1}{\beta}$ or $\beta<\frac{d-k+1}{l_2-1}$.
\end{corollary}

\begin{proof}
Combining Theorem \ref{third proof} and Theorem \ref{fourth proof}, this corollary can be derived as Corollary \ref{Secure size further1}.
\end{proof}

\begin{rem}
In this category, achievablity can be attributed to that $H(S_g^F)$ exactly reaches the maximal value $l_2\beta$, when $l_2<1+\frac{d-k+1}{\beta}$. In other words, there does not exist the intersection pattern (dependence) within $S_g^F$ in this category, i.e., all repair data included in $S_g^F$ are mutually independent. However, sometimes $H(S_g^F)$ cannot be exactly calculated and only can be estimated, which will be shown next.
\end{rem}

\subsection{Category 2: Upper Bounds on Secrecy Capacity}
In the other situations when $\beta\geq \frac{d-k+1}{l_2-1}$, we cannot exactly calculate the value of $H(S_g^F)$. Instead, we can only estimate the range of value that $H(S_g^F)$ can be taken from.

\subsubsection{4.2.1 The situations when $l_2=t+1, \frac{d-k+1}{t}\leq\beta<\frac{d-k+1}{t-1}$.}

\begin{theorem}\label{discover}
Given a linear MSR code, for $l_1+l_2\leq k-1$ and $l_2=t+1$, when $\frac{d-k+1}{t}\leq\beta<\frac{d-k+1}{t-1}$, we have
\begin{equation}
B^{(s)}=(k-l_1-l_2)\big(\alpha-\pi(\beta,l_2)\big),
\end{equation}
where $\pi(\beta,l_2)=H(S_g^F)\geq t\beta+\frac{d-k-t+1}{d-k+1}\beta$ and $t\leq d-k+1$.
\end{theorem}

\begin{proof}
It basically follows from formulas (\ref{easy3}) and (\ref{easy4}) in Theorem \ref{fourth proof}.

When $\frac{d-k+1}{t}\leq\beta<\frac{d-k+1}{t-1}$, Corollary \ref{Secure size further2} leads to that $\pi(\beta,t)=t\beta$. According to the proof of Theorem \ref{fourth proof}, for any $i\notin [k+1,k+t+1]$, we have
\begin{equation}
\left\{\begin{aligned}
&H(S_i^{[k+1,k+t]})=t\beta\\
&H(S_i^{[k+1,k+t+1]})=t\beta+\theta,
\end{aligned}\right.
\end{equation}
where $\theta\in \mathbb{Z}\cap[0,\beta]$. Because $\beta\geq\frac{d-k+1}{t}$, when setting $\theta=\beta-1$, we have $(d+1-k)(\beta-1)\geq(d-k-t+1)\beta$, from which we cannot obtain contradiction by formula (\ref{easy3}) and (\ref{easy4}). With them, we can only derive that $\frac{d-k-t+1}{d-k+1}\beta\leq\theta\leq\beta$. Thus, we obtain
\begin{equation}\label{pi bound}
\left\{\begin{aligned}
&B^{(s)}\\
&=(k-l_1-l_2)\big(\alpha-\pi(\beta,l_2)\big)\\
&\leq (k-l_1-l_2)\{\alpha-(t\beta+\frac{d-k-t+1}{d-k+1}\beta)\},
\end{aligned}\right.
\end{equation}
where the equation of the first step follows from Remark \ref{simple} and the inequality in the second step results from that $\pi(\beta,l_2)=t\beta+\theta \geq t\beta+\frac{d-k-t+1}{d-k+1}\beta$.

\end{proof}

\begin{rem}\label{syst}
Our focus in this paper is studying the secrecy capacity of MSR codes that can efficiently repair all nodes under the eavesdropper model with $F\in[1,d+1]$. Unlike Category 1, tightness of the bounds in Theorem \ref{discover} stays unclear. In \cite{Re:Rawat.A.S}, the authors considered using Zigzag code \cite{Re:I.Tamo} to construct secure MSR code that can attain the upper bound $B^{(s)}\leq (k-l_1-l_2)\big(\alpha-(2\beta-\frac{\beta}{n-k})\big)$, where $\alpha=(n-k)^{k}$ and $|F|=l_2=2$. However, Zigzag code \cite{Re:I.Tamo} is a systematic MSR code allowing efficient repair of systematic nodes only and the secure Zigzag code designed in \cite{Re:Rawat.A.S} is established on the premise that the eavesdropper gains access to the repair data of $l_2$ systematic nodes, i.e., $F\in[1,k]$.

Nevertheless, the simple and generally applicable upper bound $B^{(s)}\leq (k-l_1-l_2)\big(\alpha-H(S_g^F)\big)$ given in our Theorem \ref{sim} in fact also applies to systematic MSR codes, only requiring that $F\in[1,k]$. First, it is clear that the universal upper bound on secrecy capacity for any regenerating code $B^{(s)}\leq H(W_G|W_E,W_F)-H(S^F|W_E,W_F)$ in Lemma \ref{secure expression} is applicable to systematic MSR codes, since they still have the reconstruction property and regeneration property of systematic nodes $[1,k]$. Further due to their minimum storage feature, it can be similarly derived that $H(W_G|W_E,W_F)=H(W_G)=(k-l_1-l_2)\alpha$. Second, based on the fact that systematic MSR codes have the same parameter setting $\alpha=(d-k+1)\beta$, one can check that Lemma \ref{conditional entropy}, Lemma \ref{Secure size} and Lemma \ref{interesting1} all apply to systematic MSR codes as well. Thus, Theorem \ref{sim} can be applied to systematic MSR codes wherein $F\in[1,k]$.

In the linear MSR scenario, although we assume $F=[k+1,k+l_2]$ in the proof of Theorem \ref{third proof}, Theorem \ref{fourth proof} and Theorem \ref{discover}, they actually all are applicable to linear systematic MSR codes, because there does not restrict $F$ to be necessarily included in $[k+1,d+1]$ in their conditions. To this end, systematic MSR codes are supposed to formally share the same secrecy capacity with MSR codes that efficiently repair all nodes. Consequently, the bound in Theorem \ref{discover} also applies to linear systematic MSR codes and actually is consistent with the bound given in \cite{Re:Rawat.A.S} for certain situation.

The secure Zigzag code present in \cite{Re:Rawat.A.S} is designed by MRD code's pre-coding (Gabidulin code \cite{Re:Gabidulin}) and is built on $\alpha=(n-k)^{k}$ and $l_2=2$. For Zigzag codes, when a systematic node is failed, the remaining $k-1$ systematic nodes and all the $n-k$ parity nodes are required to participate in repair, which implies that $d=n-1$. Thus, we have that $\alpha=(d-k+1)^{k}$ and $\beta=(d-k+1)^{k-1}\geq d-k+1$. According to our Theorem \ref{discover}, we find $t=1$ satisfies the condition as $\beta=(d-k+1)^{k-1}\geq d-k+1$, which results in that $\pi(\beta,2)\geq \beta+\frac{d-k}{d-k+1}\beta=2\beta-\frac{\beta}{d-k+1}$. It exactly equals to $2\beta-\frac{\beta}{n-k}$, the corresponding result of Corollary 16 given in \cite{Re:Rawat.A.S}. Furthermore, Corollary 16 in \cite{Re:Rawat.A.S} is apparently included in our Theorem \ref{discover}, since it is not only applicable to the situation $t=1$.

\end{rem}

\subsubsection{4.2.2 The situations when $l_2=t+e, e\geq1, \frac{d-k+1}{t}\leq\beta<\frac{d-k+1}{t-1}$.}

\begin{theorem}\label{discover1}
Given a linear MSR code, for $l_1+l_2\leq k-1$ and $l_2=t+e$, when $\frac{d-k+1}{t}\leq\beta<\frac{d-k+1}{t-1}$, we have
\begin{equation}\label{expl}
B^{(s)}=(k-l_1-l_2)(\alpha-\pi(\beta,l_2)),
\end{equation}
where $\pi(\beta,l_2)=H(S_g^F)\geq t\beta+\beta(d-k-t+1)\big{[}1-(\frac{d-k}{d-k+1})^{e} \big{]}$ with $t\leq d-k+1$ and $e\geq 1$.
\end{theorem}

\begin{proof}
Without loss of any generality, we assume the set $F$ is $[1,t+e]$, where $t+e+l_1\leq k-1$. According to Lemma \ref{interesting1}, we know that for any $i\notin[1,t+e]$, $H(S_i^{[1,t+e]})$ is invariant.

Due to $\frac{d-k+1}{t}\leq\beta<\frac{d-k+1}{t-1}$ and Corollary \ref{Secure size further2}, we have
\begin{equation}\label{exp}
\left\{\begin{aligned}
&H(S_i^{[1,t+e]})\\
&=t\beta+H(S_i^{t+1}|S_i^{[1,t]})+\cdots+H(S_i^{t+e}|S_i^{[1,t+e-1]})\\
&=t\beta+\theta_1+\cdots+\theta_e,
\end{aligned}\right.
\end{equation}
where $H(S_i^{t+j}|S_i^{[1,t+j-1]})=\theta_j$ and $\theta_j\in\mathbb{Z}\cap[0,\beta]$, for $j\in[1,e]$. Still by Lemma \ref{expre}, $H(S^{[1,t+e]})$ can be expressed as
\begin{equation}\label{easy6}
H(S^{[1,t+e]})=H(S_{[2,d+1]}^1,S_{[3,d+1]}^2,\cdots,S_{[t+e+1,d+1]}^{t+e}).
\end{equation}

\vspace{0.2cm}

First, with the method similar to the proof of Theorem \ref{third proof} and \ref{fourth proof}, we have
\begin{equation}\label{easy7}
\left\{\begin{aligned}
&H(S_{[t+e+1,d+1]}^{t+e}|S^{[1,t+e-1]})\\
&=H(S^{[1,t+e]})-H(S^{[1,t+e-1]})\\
&=\{(t+e)\alpha+(k-t-e)\big{[}t\beta+\theta_1+\cdots+\theta_{e}\big{]}\}-\{(t+e-1)\alpha+(k-t-e+1)\big{[}t\beta+\theta_1+\cdots+\theta_{e-1}\big{]}\}\\
&=\alpha-\big{[}t\beta+\theta_1+\cdots+\theta_{e-1}\big{]}+(k-t-e)\theta_e.
\end{aligned}\right.
\end{equation}

Second, we obtain
\begin{equation}\label{easy8}
\left\{\begin{aligned}
&H(S_{[t+e+1,d+1]}^{t+e}|S^{[1,t+e-1]})\\
&\leq(d-t-e+1)\theta_e,
\end{aligned}\right.
\end{equation}
which can be derived as inequality (\ref{easy4}).

\vspace{0.2cm}

Then, combining equation (\ref{easy7}) with (\ref{easy8}), we derive
$(d-k+1)\theta_e\geq \alpha-\big{[}t\beta+\theta_1+\cdots+\theta_{e-1}\big{]}$,
from which we further have
\begin{equation}
\theta_e \geq \frac{(d-k-t+1)\beta}{d-k+1}-\frac{\theta_1+\cdots+\theta_{e-1}}{d-k+1}.
\end{equation}
Through rearrangement, it can be changed to
\begin{equation}
\theta_1+\cdots+\theta_e \geq \frac{(d-k-t+1)\beta}{d-k+1}+\frac{(d-k)(\theta_1+\cdots+\theta_{e-1})}{d-k+1}.
\end{equation}
By setting $\omega(e)=\theta_1+\cdots+\theta_e$, we obtain
\begin{equation}
\omega(e) \geq \frac{d-k}{d-k+1}\omega(e-1)+\frac{(d-k-t+1)\beta}{d-k+1}.
\end{equation}
From Theorem \ref{discover}, we know $\theta_1\geq \frac{d-k-t+1}{d-k+1}\beta$. Hence, by the method of recursion and induction, we have
\begin{equation}
\omega(e) \geq \beta(d-k-t+1)\big{[}1-(\frac{d-k}{d-k+1})^e\big{]}.
\end{equation}

To this end, we have $\pi(\beta,l_2)=H(S_i^{[1,t+e]})=t\beta+\omega(e)\geq t\beta+\beta(d-k-t+1)\big{[}1-(\frac{d-k}{d-k+1})^e\big{]}$.

\end{proof}

\begin{rem}
Theorem \ref{discover1} is the supplementary of Theorem \ref{discover}, which expands the range of values that $l_2$ can be taken from. In Theorem \ref{discover}, $e$ only can be taken by $1$, while Theorem \ref{discover1} takes $e$ by any value only needing to satisfy $l_1+t+e\leq k-1$ and $t\leq d-k+1$. As stated in Remark \ref{syst}, Theorem \ref{discover1} basically also applies to systematic MSR codes for $F\in[1,k]$.

In fact, the upper bound given in \cite{Re:S.Goparaju} is also a special case of our Theorem \ref{discover1}. Zigzag codes \cite{Re:I.Tamo} by pre-coding of MRD codes are shown to be able to achieve this bound on secrecy capacity in \cite{Re:S.Goparaju}. Since $\alpha=(d+1-k)^k$, we know $\beta=(d+1-k)^{k-1}> d+1-k$, which similarly indicates that only $t=1$ conforms the constraints on $\beta$ required by our Theorem \ref{discover1}. Thereby, we have $l_2=e+1$, from which we derive $\pi(\beta,l_2)\geq\beta+\beta(d-k)\big{[}1-(\frac{d-k}{d-k+1})^{l_2-1} \big{]}$. By simplification, we further have
\begin{equation}
\left\{\begin{aligned}
&\pi(\beta,l_2)\\
&\geq \beta+\beta(d-k)\big{[}1-(\frac{d-k}{d-k+1})^{l_2-1} \big{]}\\
&=\alpha-\beta(d-k)(\frac{d-k}{d-k+1})^{l_2-1}\\
&=\alpha-\alpha(1-\frac{1}{d-k+1})^{l_2},
\end{aligned}\right.
\end{equation}
which leads to that $B^{(s)}\leq (k-l_1-l_2)\big{(}1-\frac{1}{d-k+1}\big{)}^{l_2}\alpha$. It is exactly consistent with the bound given in \cite{Re:S.Goparaju}, i.e., $B^{(s)}\leq (k-l_1-l_2)\big{(}1-\frac{1}{n-k}\big{)}^{l_2}\alpha$.

Although some upper bounds (limited to $t=1$ or $\beta>d-k+1$) in this category are achievable for the Zigzag codes considered in \cite{Re:S.Goparaju,Re:Rawat.A.S}, they are not generally achievable for all other MSR codes with $\{\beta>d-k+1\}$. For example, for those vector MSR codes with $\{\beta>d-k+1\}$ designed by concatenating $m$ same scalar MSR codes where $m>d-k+1$, their secrecy capacity is exactly equal to $(k-l_1-l_2)(\alpha-l_2\beta)$ following from Corollary \ref{Secure size further1} where $\beta=m$, since each scalar MSR code within a vector MSR code shares the same code construction and can be designed to be mutually independent. It is obvious\footnote{In the category where $l_2=t+e$ and $\frac{d-k+1}{t}\leq\beta<\frac{d-k+1}{t-1}$, we have
$B^{(s)}\leq (k-l_1-l_2)\big\{\alpha-t\beta-\beta(d-k-t+1)\big{[}1-(\frac{d-k}{d-k+1})^{e} \big{]}\big\}$. Through analysis, one can check the term $\beta(d-k-t+1)\big{[}1-(\frac{d-k}{d-k+1})^{e}\big{]}<\beta(d-k-t+1)\big{[}e(1-\frac{d-k}{d-k+1})\big{]}=e\beta(1-\frac{t}{d-k+1})=e\beta-\frac{et\beta}{d-k+1}\leq e\beta-e< e\beta$. So, we derive $t\beta+\beta(d-k-t+1)\big{[}1-(\frac{d-k}{d-k+1})^{e} \big{]}<(t+e)\beta=l_2\beta$, which leads to that $(k-l_1-l_2)(\alpha-l_2\beta)<(k-l_1-l_2)\big{(}1-\frac{1}{d-k+1}\big{)}^{l_2}\alpha$ when $t=1$.} that $(k-l_1-l_2)(\alpha-l_2\beta)<(k-l_1-l_2)\big{(}1-\frac{1}{d-k+1}\big{)}^{l_2}\alpha$, which means that those vector MSR codes cannot reach the bounds in Theorem \ref{discover1}. Therefore, unlike Category 1, the value of $H(S_g^F)$ or $\pi(\beta,l_2)$ in Category 2 cannot be determined, i.e. its precise value may vary with different MSR codes' constructions.

\end{rem}

\subsection{Putting Together}
Now combining the two categorizes on secrecy capacity of linear MSR codes, we give the following comprehensive and explicit result on secrecy capacity for any linear MSR codes with $\{n=d+1\}$.

\begin{theorem}\label{discover2}
Given a linear MSR code with $\{n=d+1,k,d,\alpha,\beta\}$, for $l_1+l_2\leq k-1$, we have
\begin{equation}
B^{(s)}=(k-l_1-l_2)\big(\alpha-\pi(\beta,l_2)\big),
\end{equation}
where $\pi(\beta,l_2)$
\begin{equation}\label{explicit}
\left\{\begin{array}{ll}
=l_2\beta, &\textrm{if} \quad l_2\leq t, \beta<\frac{d-k+1}{t-1};\\
\geq t\beta+\beta(d-k-t+1)\big{[}1-(\frac{d-k}{d-k+1})^{e} \big{]}, &\textrm{if} \quad l_2=t+e, \frac{d-k+1}{t}\leq\beta<\frac{d-k+1}{t-1},
\end{array}\right.
\end{equation}
where $1\leq t\leq d-k+1$ and $e\geq 1$.
\end{theorem}

\begin{rem}
In the literature, the known linear MSR codes are comprised of the scalar MSR codes with $\{\beta=1\}$ \cite{Re:K.Rashmi,Re:C.Suh,Re:Y.Wu1,Re:K.V.Rashmi,Re:Rashmi,Re:N. B. Shah} and the vector MSR codes with $\{\beta=(d-k+1)^{x}\}$ where $x\geq 1$ \cite{Re:Z. Wang,Re:D. S. Papailiopoulos,Re:I.Tamo,Re:Z. Wang1,Re:V. R. Cadambe,Re:V. R. Cadambe1,Re:G. K. Agarwal,Re:Y.S.Han,Re:J.Li}. It should be noted that these vector MSR codes are not designed from concatenation of scalar MSR codes. Similar to Zigzag codes \cite{Re:I.Tamo}, they share the same intersection pattern, i.e. there exist the same dependence within disjoint sets of repair data transmitted from any one node (e.g. $S_g^F$ where $g\notin F$). Thus, the second item in formula (\ref{explicit}) also applies to them.

\end{rem}

\subsection{Further Discussions}

Theorem \ref{discover2} exhibits a comprehensive and explicit result on secrecy capacity for any linear MSR code given the parameter set $\{n=d+1,k,d,\alpha,\beta\}$ and the $(l_1,l_2)$-eavesdropper model. In retrospect, all constructions of linear MSR codes are based on the scalar case $\beta=1$ or partial vector cases where $\beta$ is required to be exponential in $d-k+1$. Thus, designing linear vector MSR codes with $\{1<\beta<d-k+1\}$ by no concatenation remains open. Nevertheless, our Theorem \ref{discover2} still presents certain results on secrecy capacity for these unexplored MSR codes. Thereby, we put forward two questions as follows.

\begin{question}
Do there exist the linear MSR codes with $\{1<\beta<d-k+1\}$ by no concatenation? If so, how can we construct them?
\end{question}

\begin{question}
Given such a construction with $\{1<\beta<d-k+1\}$, is it achievable for the bounds given in formula (\ref{explicit}) when $l_2\geq 1+\frac{d-k+1}{\beta}$?
\end{question}

\begin{rem}
According to formula (\ref{explicit}), for the linear MSR codes with $\{1<\beta<d-k+1\}$ by no concatenation, when $l_2< 1+\frac{d-k+1}{\beta}$, it must be that
\begin{equation}\label{unex1}
B^{(s)}=(k-l_1-l_2)(\alpha-l_2\beta).
\end{equation}
However, when $l_2\geq 1+\frac{d-k+1}{\beta}$, formula (\ref{explicit}) leads to that
\begin{equation}\label{unex}
B^{(s)}\leq(k-l_1-l_2)\big(\alpha-t\beta-\beta(d-k-t+1)\big{[}1-(\frac{d-k}{d-k+1})^{l_2-t} \big{]}\big),
\end{equation}
where $t=\frac{d-k+1}{\beta}$ if $\beta$ divides $d-k+1$, and $t=\lfloor 1+\frac{d-k+1}{\beta} \rfloor$ if $\beta$ does not divide $d-k+1$. Hence, as in Question 2, we ask whether it is achievable for the upper bound (\ref{unex}) given such a code.

Overall, Theorem \ref{discover2} predicts certain results on secrecy capacity for these unexplored MSR codes, which consist of the precise value of secrecy capacity (\ref{unex1}) and the upper bound on secrecy capacity (\ref{unex}).
\end{rem}

\section{CONCLUSION}

In this paper, we carry out research on the secrecy capacity of MSR codes. We assume the passive adversarial model where the eavesdropper can observe the contents of certain nodes and the repair data of some other nodes. Although the secrecy capacity of MBR codes has been characterized completely \cite{Re:S.Pawar}, it is a challenging task to analyze the secrecy capacity of MSR codes \cite{Re:N.B.Shah,Re:S.Goparaju,Re:Rawat.A.S}. The additional difficulty comes from the fact that the amount of data transmitted for a failed node in MSR codes, is not entirely stored on the node undergoing repair, making it challenging to compute the joint entropy of the repair data. With such a system model, we focus on investigating the repair data in the MSR scenario from the information-theoretic perspective.

We first obtain some information-theoretic properties and some upper bounds on secrecy capacity for general MSR codes, in addition to which we introduce a new concept named by \emph{stable} MSR codes. For the unstable MSR codes, we assume the eavesdropper could identify the nodes with repair data captured and demonstrate that its secrecy capacity is strictly less than that of \emph{stable} MSR code. In the linear MSR scenario, we utilize permutation polynomial and orthogonal system in finite fields to explain the fact that entropy of multiple sets of repair data is an integer and ultimately derive a comprehensive and explicit result on secrecy capacity which closely depends on the value of $\beta$. This outcome not only explains and generalizes the previous results in \cite{Re:N.B.Shah,Re:S.Goparaju,Re:Rawat.A.S}, but also predicts certain results for some unexplored linear MSR codes. After that, we put forward two related questions. On the other hand, we find that all of these results also apply to systematic MSR codes with repair data of systematic nodes captured.

\end{document}